\documentclass[11pt,onecolumn]{article}
\usepackage{amsmath,amsthm,amssymb,amsfonts,epsfig,graphicx}
\usepackage{algorithm}
\usepackage{algorithmic}

\usepackage{color}
\usepackage{dsfont}
\usepackage{wrapfig}
\usepackage{relsize}
\usepackage{IEEEtrantools}
\usepackage{url}

\usepackage[left=1.0in,top=1.0in,right=1.0in,bottom=1.0in,nohead,textheight=10in,footskip=0.3in]{geometry}

\newtheorem{theorem}{Theorem}

\newtheorem{lemma}[theorem]{Lemma}
\newtheorem{definition}[theorem]{Definition}
\newtheorem{proposition}[theorem]{Proposition}
\newtheorem{remark}{Remark}

\begin{document}

\def\myparagraph#1{\vspace{2pt}\noindent{\bf #1~~}}

\newcounter{fooTH}
\newcounter{fooEQ}

\def\RA#1{{  \;\stackrel{#1}{\rightarrow}\;   }}
\def\shortldots{{\hspace{1pt}.\hspace{1pt}.\hspace{1pt}.\hspace{1pt}}}

\newcommand{\setTHcounter}[1]{\setcounter{theorem}{#1}\addtocounter{theorem}{-1}}
\newcommand{\saveTHcounter}{\setcounter{fooTH}{\value{theorem}}}
\newcommand{\restoreTHcounter}{\setcounter{theorem}{\thefooTH}}

\newcommand{\setEQcounter}[1]{\setcounter{equation}{#1}\addtocounter{equation}{-1}}
\newcommand{\saveEQcounter}{\setcounter{fooEQ}{\value{equation}}}
\newcommand{\restoreEQcounter}{\setcounter{equation}{\thefooEQ}}

\newcommand{\eqdef}{{\stackrel{\mbox{\tiny \tt ~def~}}{=}}}
\def\calQ{{\cal Q}}
\def\Q{{Q}}

\def\Bad{{\tt Bad}}
\def\BadPar{{\tt BadPar}}
\def\Image{{\tt Image}}
\def\TwoDeltaCeil{{\lceil 2\Delta\rceil}}
\def\OnePointFiveDeltaCeil{{\lceil 3\Delta/2\rceil}}

\long\def\ignore#1{}
%\epsfverbosetrue
%\def\myps[#1]#2{}
\def\myps[#1]#2{\includegraphics[#1]{#2}}
\def\etal{{\em et al.}}
\def\Bar#1{{\bar #1}}
\def\br(#1,#2){{\langle #1,#2 \rangle}}
\def\setZ[#1,#2]{{[ #1 .. #2 ]}}
\def\Pr{\mbox{\rm Pr}}
\def\REACHED{\mbox{\tt REACHED}}
\def\AdjustFlow{\mbox{\tt AdjustFlow}}
\def\GetNeighbors{\mbox{\tt GetNeighbors}}
\def\true{\mbox{\tt true}}
\def\false{\mbox{\tt false}}
\def\Process{\mbox{\tt Process}}
\def\ProcessLeft{\mbox{\tt ProcessLeft}}
\def\ProcessRight{\mbox{\tt ProcessRight}}
\def\Add{\mbox{\tt Add}}

\def\suchthat#1#2{\setof{\,#1\mid#2\,}} % so says Knuth, page 174
\def\event#1{\setof{#1}}
\def\q={\quad=\quad}
\def\qq={\qquad=\qquad}
\def\calA{{\cal A}}
\def\calB{{\cal B}}
\def\calC{{\cal C}}
\def\calD{{\cal D}}
\def\calE{{\cal E}}
\def\calF{{\cal F}}
\def\calG{{\cal G}}
\def\calI{{\cal I}}
\def\calH{{\cal H}}
\def\calL{{\cal L}}
\def\calM{{\cal M}}
\def\calN{{\cal N}}
\def\calP{{\cal P}}
\def\calR{{\cal R}}
\def\calS{{\cal S}}
\def\calT{{\cal T}}
\def\calU{{\cal U}}
\def\calV{{\cal V}}
\def\calO{{\cal O}}
\def\calX{{\cal X}}
\def\calY{{\cal Y}}
\def\s{\footnotesize}
\def\calNG{{\cal N_G}}
\def\psfile[#1]#2{}
\def\psfilehere[#1]#2{}
\def\epsfw#1#2{\includegraphics[width=#1\hsize]{#2}}
\def\assign(#1,#2){\langle#1,#2\rangle}
\def\edge(#1,#2){(#1,#2)}
\def\VS{\calV^s}
\def\VT{\calV^t}
\def\slack(#1){\texttt{slack}({#1})}
\def\barslack(#1){\overline{\texttt{slack}}({#1})}
\def\NULL{\texttt{NULL}}
\def\PARENT{\texttt{PARENT}}
\def\GRANDPARENT{\texttt{GRANDPARENT}}
\def\TAIL{\texttt{TAIL}}
\def\HEADORIG{\texttt{HEAD$\_\:$ORIG}}
\def\TAILORIG{\texttt{TAIL$\_\:$ORIG}}
\def\HEAD{\texttt{HEAD}}
\def\CURRENTEDGE{\texttt{CURRENT$\!\_\:$EDGE}}

\def\unitvec(#1){{{\bf u}_{#1}}}
\def\uvec{{\bf u}}
\def\vvec{{\bf v}}
\def\Nvec{{\bf N}}

\def\IND(#1){{{\tt Ind}(\Gamma({#1}))}}
\def\Ind(#1){{{\tt Ind}({#1})}}
\def\ELF{{\Psi}}
\def\sigmainit{{\sigma^{\tt init}}}
\def\omegainit{{\omega^{\tt init}}}
\def\gammainit{{\gamma^{\tt init}}}
\def\Indinit{{{\tt Ind}^{\tt init}}}

\newcommand{\bg}{\mbox{$\bf g$}}
\newcommand{\bh}{\mbox{$\bf h$}}

\newcommand{\bx}{\mbox{$x$}}
\newcommand{\by}{\mbox{\boldmath $y$}}
\newcommand{\bz}{\mbox{\boldmath $z$}}
\newcommand{\bu}{\mbox{\boldmath $u$}}
\newcommand{\bv}{\mbox{\boldmath $v$}}
\newcommand{\bw}{\mbox{\boldmath $w$}}
\newcommand{\bvarphi}{\mbox{\boldmath $\varphi$}}
\newcommand{\bOmega}{\mbox{\boldmath $\Omega$}}
\newcommand{\bsigma}{\mbox{\boldmath $\sigma$}}

\newcommand\myqed{{}}

\def\R#1{{  \;\stackrel{#1}{\rightarrow}\;   }}
\def\SETOF#1{{\langle{#1}\rangle}}

\def\myeq{{\,\stackrel{\bullet}{=}\,}}

%%%%%%%%%%%%%%%%%%%%%%%%%%%%%%%%%%%%%%%%%%%%%%%%%%%%%%%%%%%
%%%%%%%%%%%%%%%%%%%%%%%%%%%%%%%%%%%%%%%%%%%%%%%%%%%%%%%%%%%
%%%%%%%%%%%%%%%%%%%%%%%%%%%%%%%%%%%%%%%%%%%%%%%%%%%%%%%%%%%

\title{
\Large\bf  
Commutativity in the Algorithmic Lov\'asz Local Lemma}
\author{Vladimir Kolmogorov \\ \normalsize Institute of Science and Technology Austria \\ {\normalsize\tt vnk@ist.ac.at}}
\date{}
\maketitle

% REQUIRED
\begin{abstract}
We consider the recent formulation of 
the Algorithmic Lov\'asz Local Lemma~\cite{HarveyVondrak15,Achlioptas:new,LLL:submitted} for finding objects that avoid ``bad features'', or ``flaws''.
It extends the Moser-Tardos resampling algorithm~\cite{MoserTardos} to more general discrete spaces.
At each step the method picks a flaw present
in the current state and goes to a new state according to some prespecified probability distribution (which
depends on the current state and the selected flaw).
However, it is less flexible than the Moser-Tardos method
since \cite{HarveyVondrak15,Achlioptas:new,LLL:submitted} require a specific flaw selection rule, whereas \cite{MoserTardos} allows
an arbitrary rule (and thus can potentially be implemented more efficiently).

We formulate a new ``commutativity'' condition, and prove that it is sufficient
for an arbitrary rule to work. It also enables an efficient parallelization under an additional
assumption.
We then show that existing resampling oracles
for perfect matchings and permutations do satisfy this condition.
%
%Finally, we generalize the precondition in~\cite{Achlioptas:new} (in the case of symmetric potential causality graphs). This unifies special cases that previously were treated separately.
\end{abstract}

\section{Introduction}
Let $\Omega$ be a (large) set of objects and $F$ be a set of {\em flaws},
where a flaw $f\in F$ is some non-empty set of ``bad'' objects, i.e.\ $f\subseteq\Omega$. 
Flaw $f$ is said to be {\em present in $\sigma$} if $\sigma\in f$.
Let $F_\sigma=\{f\in F\:|\:\sigma\in f\}$ be the set of flaws present in $\sigma$.
Object $\sigma$ is called {\em flawless} if $F_\sigma=\varnothing$.

The existence of flawless objects can often be shown via a probabilistic method.
First, a probability measure $\omega$ on $\Omega$ is introduced, then
flaws in $F$ become (bad) events that should be avoided.
Proving the existence of a flawless object is now equivalent to
showing that the probability of avoiding all bad events is %strictly
 positive.
This holds if, for example, all events $f\in F$ are independent and the probability of each $f$ is smaller than 1.
The well-known Lov\'asz Local Lemma (LLL)~\cite{LLL75} is a powerful tool that can handle a (limited) dependency between the events.
Roughly speaking, it states that if the dependency graph is sparse enough (e.g.\ has a bounded degree)
and the probabilities of individual bad events are sufficiently small then a flawless object is guaranteed to exist.

LLL has been the subject of intensive research, see e.g.\ \cite{Szegedy:survey}
for a relatively recent survey. One of the milestone results
was the {\em constructive} version of LLL by Moser and Tardos~\cite{MoserTardos}.
It applies to the {\em variable model} in which $\Omega=\calX_1\times\ldots\times \calX_n$
for some discrete sets $\calX_i$, event $f$ depends on a small subset of variables denoted as ${\tt vbl}(f)\subseteq [n]$,
and two events $f,g$ are declared to be dependent if ${\tt vbl}(f)\cap{\tt vbl}(g)\ne \varnothing$.
The algorithm proposed in~\cite{MoserTardos} is strikingly simple:
(i) sample each variable $\sigma_i$ for $i\in[n]$ according to its
distribution; (ii) while $F_\sigma$ is non-empty, pick an arbitrary flaw $f\in F_\sigma$
and resample all variables $\sigma_i$ for $i\in{\tt vbl}(f)$. Moser and Tardos proved that if the LLL condition in~\cite{LLL75}
is satisfied then the expected number of resamplings is finite (polynomial for most of the known applications).

The recent development has been extending algorithmic LLL beyond the variable 
model, and in particular to non-Cartesian spaces. %~\cite{Harris:permutations,Achlioptas,HarveyVondrak15,Achlioptas:new}.
The first such work was by Harris and Srinivasan~\cite{Harris:permutations},
who considered the space of permutations. Achlioptas and Iliopoulos~\cite{Achlioptas}
introduced a more abstract framework where the behaviour of the algorithm is specified by a certain multigraph.
Harvey and Vondr\'ak~\cite{HarveyVondrak15} introduced the notion of {\em resampling oracles},
providing an algorithmization of the LLL given
efficiently implementable resampling oracles.
They also characterized the condition under which a resampling oracle exists:
it was called  {\em lopsided association} in~\cite{HarveyVondrak15}, and
was shown  to lie strictly between the original asymmetric LLL condition~\cite{LLL75}
and a more refined {\em lopsidependency} LLL condition~\cite{ErdosSpencer:LLLL}.
%As shown in~\cite{HarveyVondrak15},
%a resampling oracle exist under a certain {\em lopsided association}
%condition, which lies strictly between the original LLL condition~\cite{LLL75}
%and a more refined {\em lopsidependency} LLL condition~\cite{ErdosSpencer:LLLL}.

By definition, resampling oracles must satisfy a certain property intimately tying them
to the LLL measure~$\omega$.
Achlioptas and Iliopoulos~\cite{Achlioptas:new} gave a first analysis of LLL-inspired
algorithms whose transition probabilities do not need to satisfy this property,
and in~\cite{LLL:submitted} this analysis was simplified and extended.
%We will refer to this setting from  
% to yield a framework
%in which algorithm's compatibility with the LLL measure can be 
%
%
%proposed a framework with {\em resampling oracles}~\cite{HarveyVondrak15}, providing a more direct connection to LLL.
%These oracles were required to satisfy a certain condition that was called {\em regeneration} in~\cite{Achlioptas:new}.
%%in particular with so-called {\em regenerating} resampling oracles.
%%This framework was further generalized by
%More general oracles that go beyond regeneration
%were proposed by Achlioptas and Iliopoulos~\cite{Achlioptas:new} and by Achlioptas et al.~\cite{LLL:submitted}.
%}
%
We will refer to this setting from \cite{Achlioptas:new,LLL:submitted} as ``general algorithmic LLL''
(even though it contains cases that go beyond the probabilistic version of LLL~\cite{LLL75,ErdosSpencer:LLLL}).

This is the setting studied in this paper.
It does not assume any particular structure on sets $\Omega$ and $F$.
Instead, for each object $\sigma\in\Omega$ and flaw $f\in F_\sigma$ the user must provide
an oracle that will used for sampling a new object. It is
%which is a probability distribution ${\mathbb P}^f_\sigma$ over $\Omega$.
 specified by a set of {\em actions} $A(f,\sigma)\subseteq \Omega$ that can be taken to ``address'' flaw $f$,
and a probability distribution $\rho(\sigma'|f,\sigma)$ over $\sigma'\in A(f,\sigma)$.
At each step the algorithm selects a certain flaw $f\in F_\sigma$, samples an action $\sigma'\in A(f,\sigma)$
according to $\rho(\sigma'|f,\sigma)$, and goes there. 
This framework captures the Moser-Tardos algorithm~\cite{MoserTardos},
and can also handle other scenarios such as permutations and perfect matchings
(in which case $\Omega$ cannot be expressed as a Cartesian product).

One intriguing difference between the methods of~\cite{MoserTardos} and~\cite{Achlioptas,HarveyVondrak15,Achlioptas:new,LLL:submitted}
is that~\cite{MoserTardos} allows an arbitrary rule for selecting a flaw $f\in F_\sigma$,
whereas~\cite{Achlioptas,HarveyVondrak15,Achlioptas:new,LLL:submitted} require a specific rule (which depends on a permutation $\pi$
of $F$ chosen in 
advance)\footnote{The papers~\cite{Achlioptas,Achlioptas:new} actually allowed more freedom in the choice of permutation $\pi$, e.g.\ it
may depend on the iteration number. However, once $\pi$ has been chosen, the algorithm should
still  examine some ``current'' set of flaws and pick the lowest one with respect to $\pi$.}. We will say that an algorithm is {\em flexible}
if it is guaranteed to work with any flaw selection rule.
We argue that flexibility can lead to a much more efficient practical implementation:
 it is not necessary to examine all flaws in $F_\sigma$,
the first found flaw will suffice. If the  list of current flaws is updated dynamically
then flexibility could potentially eliminate the need for a costly data structure (such as a priority queue)
and thus save a factor of $\Theta(\log n)$ in the complexity.
The rule may also affect the number of steps in practice; experimentally,
the selection process matters, as noted in~\cite{Szegedy:survey}.

Achlioptas and Iliopoulos discuss flaw selection rules in~\cite[Section 4.3]{Achlioptas},
and remark that they do not see how to accommodate arbitrary rules in their framework.
It is known, however, that in special cases flexible rules can be used even beyond the variable model.
Namely, through a lengthy and a complicated analysis Harris and Srinivasan~\cite{Harris:permutations}
managed to show the correctness of a resampling algorithm for permutations,
and did not make assumptions on the flaw selection rule in their proof. They also proved a better bound for the parallel version of the algorithm.

This paper aims to understand which properties of the problem enable flexibility and parallelism.
Our contributions are as follows.
%\footnote{\color{red} The conference version of this paper~\cite{LLL:focs}
%lists another contribution, namely generalizing the condition in~\cite{Achlioptas:new} for the algorithmic LLL to work (in the case of symmetric potential causality graphs).
%For reviewers information, this part has been merged with~\cite{Achlioptas:new} upon
%request from D.~Achlioptas and F.~Iliopoulos, and will be submitted elsewhere~\cite{LLL:submitted}.
%}
\begin{itemize}\setlength\itemsep{0pt}
\item We formulate a new condition that we call ``commutativity'', and prove that it is sufficient for flexibility.
\item We prove that it gives a better bound on the number of rounds of the parallel version of the algorithm.
In particular, we show how to use commutativity for handling ``partial execution logs'' instead of ``full execution logs''
(which is required for analyzing the parallel version).
\item We show that existing resampling oracles for permutations~\cite{Harris:permutations}
and perfect matchings in complete graphs~\cite{HarveyVondrak15} are commutative. (In fact, we treat
both cases in a single framework). Thus, we provide a simpler proof of the result in~\cite{Harris:permutations}
and generalize it to other settings, in particular to perfect matchings in certain graphs
(for which existing algorithms require specific rules). 
%\item We generalize the condition in~\cite{Achlioptas:new} for the algorithmic LLL to work (in the case of symmetric potential causality graphs).
%The new condition unifies special cases that were treated separately in~\cite{Achlioptas:new}.
\end{itemize}
To our knowledge, our commutativity condition captures all previously known cases when the flaw selection rules was allowed to be arbitrary.
%Thus, we argue that commutativity is the underlying reason for flexibility and parallelism.
 
\myparagraph{Other related work}
Applications that involve non-Cartesian spaces $\Omega$ (such as permutations, matchings
and spanning trees) have often been tackled via the  {\em Lopsided LLL}~\cite{ErdosSpencer:LLLL};
we refer to~\cite{Lu:13,Mohr:PhD} for a comprehensive survey.
On the level of techniques there is some connection between this paper and a recent work by Knuth~\cite{Knuth};
we discuss this in Section~\ref{sec:results}.

%Interestingly, Harvey and Vondr\'ak make a link between the existence of the resampling oracle
%and conditions of the {\em Lopsided LLL}~\cite{ErdosSpencer:LLLL}. In contrast,
%there is no known explicit connection between the framework of~\cite{Achlioptas} and LLL
%(except for special cases, such as the variable setting).

\section{Background and preliminaries}
First, we give a formal description of the general algorithmic LLL framework. %algorithm from~\cite{HarveyVondrak15,Achlioptas:new}.
As described in the previous section,
we assume that for each object $\sigma\in\Omega$ and each flaw $f\in F_\sigma$
there is a non-empty set of  actions $A(f,\sigma)\subseteq \Omega$ that can be taken
 for ``addressing'' flaw $f$ at $\sigma$,
and a probability distribution $\rho(\sigma'|f,\sigma)$ over $\sigma'\in A(f,\sigma)$.
Note, by definition $A(f,\sigma)$ is the support of distribution $\rho(\cdot|f,\sigma)$.
The collection of all these distributions will be denoted as $\rho$.
We fix some probability distribution  $\omega$ on $\Omega$ with $\omega(\sigma)>0$ for all $\sigma\in\Omega$
(it will be used later for formulating various conditions). Note that our notation
is quite different from that of Harvey and Vondr\'ak~\cite{HarveyVondrak15}.\footnote{
``Flaws'' $f$ correspond to ``bad events'' $E_i$ in~\cite{HarveyVondrak15}.
The distribution over $\Omega$ was denoted in~\cite{HarveyVondrak15} as $\mu$, the states of $\Omega$ as $\omega$,
and the resampling oracle for the bad event $E_i$ at state $\omega\in\Omega$ as $r_i(\omega)$. 
}
The algorithm can now be stated as follows.
%It performs a random walk by moving between objects in $\Omega$,
%and so it is natural to call objects $\sigma\in\Omega$ ``states''.

\begin{algorithm}[!h]
\caption{Random walk. Input: initial distribution $\omegainit$ over $\Omega$, strategy $\Lambda$.}\label{alg}
\begin{algorithmic}[1]
\STATE sample $\sigma\in\Omega$ according to $\omegainit$
\STATE {\bf while} $F_{\sigma}$ non-empty {\bf do}
\STATE ~~~~ select flaw $f\in F_{\sigma}$ according to $\Lambda$
\STATE ~~~~ sample $\sigma'\in A(f,\sigma)$ according to distribution $\rho(\sigma'|f,\sigma)$, set $\sigma\leftarrow\sigma'$.
\STATE {\bf end while}
%\STATE set $i=1$
%\STATE {\bf while} $F_{\sigma_i}$ non-empty {\bf do}
%\STATE ~~~~ select flaw $w_i\in F_{\sigma_i}$ according to $\Lambda$
%\STATE ~~~~ sample $\sigma_{i+1}\in A(\sigma_i, w_i)$ uniformly at random
%\STATE ~~~~ set $i\leftarrow i+1$.
%\STATE {\bf end while}
\end{algorithmic}
\end{algorithm}

Clearly, if the algorithm terminates then it produces a flawless object $\sigma$.
The works~\cite{Achlioptas,HarveyVondrak15,Achlioptas:new,LLL:submitted} used specific strategies $\Lambda$.
As stated in the introduction, our goal is to understand when an arbitrary strategy can be used.
%In contrast, we will be interested in the case when $\Lambda$ can be arbitrary. 
This means that flaw $f$ in line 3 is selected
according to some distribution which is a function of the entire past
execution history\footnote{The description of the algorithm in~\cite{MoserTardos}
says ``{\em pick an arbitrary violated event}''. This is consistent with
our definition of an ``arbitrary strategy'': in the analysis Moser and Tardos 
mention that this selection must come from some fixed procedure (either deterministic
or randomized), so that expected values are well-defined.}.
Note that if flaw $f\in F_\sigma$ in line 3 depends only on $\sigma$
then the algorithm can be viewed as a random walk in a Markov chain with states $\Omega$,
while in a more general case the walk can be non-Markovian. 

\subsection{Walks and the potential causality graph}

We say that $\sigma\stackrel{f}\rightarrow\sigma'$ is a (valid) walk
if it is possible to get from state $\sigma$  to  $\sigma'$ by ``addressing'' flaw $f$
as described in the algorithm,
i.e.\ if two conditions hold: $f\in F_\sigma$ and $\sigma'\in A(f,\sigma)$. 
Whenever we write $\sigma\stackrel{f}\rightarrow\sigma'$, we mean that it is a valid walk. 

In many applications sampling oracles satisfy a special condition called {\em atomicity}~\cite{Achlioptas}.
%Let $D$ be a multigraph with nodes $\Omega$ whose set of edges is
%the set of all walks $\sigma\stackrel{f}\rightarrow\sigma'$.
%(Each edge of $D$ is labeled by a flaw in $F$). Note that $D$ completely
%specifies the input data, i.e.\ sets $\Omega$, $F$ and $A(f,\sigma)$ for $f\in F$, $\sigma\in\Omega$.
%As in~\cite{Achlioptas}, we assume that $D$ is {\em atomic}.
\begin{definition}\label{def:atomic}
$\rho$ is called {\em atomic} if for any $f\in F$
and $\sigma'\in\Omega$ there exists at most one object $\sigma\in\Omega$ such that $\sigma\stackrel{f}\rightarrow\sigma'$. 
\end{definition}
%Achlioptas and Iliopoulos argue that atomicity is a very natural requirement,
%and remark that it was present in all applications that they have considered.
%For example, in the variable model of Moser and Tardos the atomicity will hold
%if each flaw $f$ is {\em atomic}, i.e.\ has the following form:
%\begin{itemize}
%\item {\em
%There exists subset $U={\tt vbl}(f)\subseteq[n]$ and a partial assignment $\bar\sigma\in\times_{i\in U}\calX_i$
%such that $f=\{\sigma\in\Omega\:|\:\mbox{the restriction of $\sigma$ to $U$ equals $\bar\sigma$}\}$.
%}
%\end{itemize}
%(Here it is assumed that $A(f,\sigma)$ for $\sigma\in f$ is the set of all objects $\sigma'\in\Omega$
%that agree with $\sigma$ on variables $i\in[n]-{\tt vbl}(f)$).
%Restricting to atomic flaws does not reduce the modeling power since 
% any flaw $f\subseteq\Omega=\calX_1\times\ldots\times\calX_n$ can be represented as a union of atomic flaws.

Next, we need to describe ``dependences'' between flaws in $F$.
Let $\sim$ be some symmetric relation on $F$ (so that $(F,\sim)$ is an undirected graph).
%Resampling oracles $\rho$ and relation $\sim$ 
It is assumed to be fixed throughout the paper.
For a flaw $f\in F$ let $\Gamma(f)=\{g\in F\:|\:f\sim g\}\s$ be the set
of neighbors of $f$. Note, we may or may not have $f\sim f$,
and so $\Gamma(f)$ may or may not contain $f$. We will denote $\Gamma^+(f)=\Gamma(f)\cup\{f\}$,
and also $\Gamma(S)=\bigcup_{f\in S}\Gamma(f)$ and $\Gamma^+(S)=\bigcup_{f\in S}\Gamma^+(f)$ for a subset
 $S\subseteq F$.

\begin{definition}\label{def:causality}
Undirected graph $(F,\sim)$ is called a {\em potential causality graph} for $\rho$
if for any walk $\sigma\stackrel{f}\rightarrow\sigma'$ there holds $F_{\sigma'}\subseteq(F_{\sigma}-\{f\})\cup\Gamma(f)$.
\end{definition}
In other words, $\Gamma(f)$ must contain all flaws that can appear after addressing flaw $f$ at some state.
Also, $\Gamma(f)$ must contain $f$ if addressing $f$ at some state can fail to eradicate $f$.

%Also note that in Definition~\ref{def:causality} 
Note that in Definition~\ref{def:causality} we deviate slightly from~\cite{Achlioptas,Achlioptas:new,LLL:submitted}: in their analysis the potential causality graph was {\em directed}
and therefore in certain cases could capture more information.
While directed graphs do matter in some applications (see examples in~\cite{Achlioptas,Achlioptas:new}),
we believe that in a typical application the potential causality relation is symmetric.
Using an undirected graph will be essential for incorporating commutativity.

\ignore{
The potential causality graph should not be confused with the {\em dependency graph}
used in the (lopsided) LLL. The latter describes conditional dependencies between random events,
and thus has a different meaning. %Since two frameworks use a different 
For this reason there is no explicit connection between the framework of~\cite{Achlioptas}
and any particular version of LLL (even though in many applications the graphs turn out to have an identical structure).
}

A subset $S\subseteq F$ will be called {\em independent} 
if for any {\em distinct} $f,g\in S$ we have $f\nsim g$. (Thus, loops $f\sim f$ in the graph $(F,\sim)$
do not affect the definition of independence).
For a subset $S\subseteq F$ we denote $\Ind(S)=\{T\subseteq S\:|\:T\mbox{ is independent}\}$.

\subsection{Commutativity} 
We now formulate  new conditions that will allow an arbitrary flaw selection rule to be used.
\begin{definition}\label{def:main}
$(\rho,\sim)$ is called {\em weakly commutative} if there exists a mapping ${\tt SWAP}$
that sends any walk 
 $\sigma_1\RA{f}\sigma_2\RA{g}\sigma_3$
with $f\nsim g$ to another valid walk 
$\sigma_1\RA{g}\sigma'_2\RA{f}\sigma_3$, and this mapping
is injective.
\end{definition}
Note that in the atomic case the definition can be simplified. Namely, $(\rho,\sim)$ is weakly commutative if and only if it satisfies the following condition:
%then the condition is equivalent to the following:
\begin{itemize}
\item {\em For any walk $\sigma_1\RA{f}\sigma_2\RA{g}\sigma_3$
with $f\nsim g$ there exists state $\sigma'_2\in\Omega$ such that 
$\sigma_1\RA{g}\sigma'_2\RA{f}\sigma_3$ is also a walk. }
\end{itemize}
%Note, by atomicity the state $\sigma'_2$ is unique.
Indeed, by atomicity
the state $\sigma'_2$ is unique, and so mapping ${\tt SWAP}$
in Definition~\ref{def:main} is constructed in a natural way.
This mapping is reversible and thus injective.

For several results we will also need a stronger property.

\begin{definition}\label{main:strong}
$(\rho,\sim)$ is called {\em strongly commutative} %(or just {\em commutative}) 
if for any walk $\tau=\sigma_1\RA{f}\sigma_2\RA{g}\sigma_3$
with $f\nsim g$ and ${\tt SWAP}(\tau)=\sigma_1\RA{g}\sigma'_2\RA{f}\sigma_3$ there holds
\begin{equation}
\rho(\sigma_2|f,\sigma_1)\rho(\sigma_3|g,\sigma_2)=\rho(\sigma'_2|g,\sigma_1)\rho(\sigma_3|f,\sigma'_2)
\end{equation}
\end{definition}

It is straightforward to check that strong commutativity holds in the variable model of Moser and Tardos.
In fact, an additional property holds: for any
$\sigma_1\stackrel{f}\rightarrow\sigma_2\stackrel{g}\rightarrow\sigma_3$
with $f\nsim g$ there exists exactly one state $\sigma'_2\in\Omega$ such that 
$\sigma_1\stackrel{g}\rightarrow\sigma'_2\stackrel{f}\rightarrow\sigma_3$.
%(and so mapping ${\tt SWAP}$ in Definition~\ref{def:main} is constructed in a natural way). 
%
Checking strong commutativity for non-Cartesian spaces $\Omega$ is more involved; we refer to Section~\ref{sec:construction} for details.

\subsection{Parallel version} We will also consider the following version of the algorithm (see Algorithm~\ref{alg:parallel}).
As presented, the algorithm is actually sequential. However, 
with appropriate implementations of lines 3-8 
 it becomes equivalent to some existing parallel algorithms,
namely 
to the parallel algorithm of Moser and Tardos~\cite{MoserTardos} in the case of the variable model,
and to the parallel algorithm of 
Harris and Srinivasan~\cite{Harris:permutations} in the case
of permutations. 
Algorithm~\ref{alg:parallel} can thus be viewed as a framework for parallel algorithms.
Note that it is closely related to the ``{\tt MaximalSetResample}'' algorithm of Harvey and Vondr\'ak~\cite{HarveyVondrak15}
(see below).

\begin{algorithm}[!h]
\caption{Parallel random walk. %Input: initial object $\sigmainit\in \Omega$.
}\label{alg:parallel}
\begin{algorithmic}[1]
\STATE sampe $\sigma\in\Omega$ according to distribution $\omegainit$
\STATE {\bf while} $F_{\sigma}$ non-empty {\bf do}
\STATE ~~~~ set $I=\varnothing$
\STATE ~~~~ {\bf while} set $F_\sigma-\Gamma^+(I)$ is non-empty {\bf do}
\STATE ~~~~~~~~ pick some $f\in F_\sigma-\Gamma^+(I)$
\STATE ~~~~~~~~ sample $\sigma'\in A(f,\sigma)$ according to $\rho(\sigma'|f,\sigma)$, set $\sigma\leftarrow\sigma'$.
\STATE ~~~~~~~~ set $I\leftarrow I\cup\{f\}$
\STATE ~~~~ {\bf end while}
\STATE {\bf end while}
\end{algorithmic}
\end{algorithm}
Lines 3-8 will be called a {\em round}.
In some cases each round admits an efficient parallel implementation (with a polylogarithmic running time).
In particular, this has been shown for the variable model of Moser and Tardos~\cite{MoserTardos}
and for permutations~\cite{Harris:permutations}.
Accordingly, we will be interested in the number of rounds of the algorithm.

Note, during round $r$ the set $F_\sigma-\Gamma^+(I)$ in line 5 shrinks
from iteration to iteration (and so flaw $f$ in line 5
satisfies $f\in F_{\sigma_r}$, where $\sigma_r$ is the state in the beginning of round $r$).
This property can be easily verified using induction and Definition~\ref{def:causality}.

%Note, in line 5 of round $r$ we must have $F_\sigma-\Gamma^+(I)\subseteq F_{\sigma_r}$
%where $\sigma_r$ is the state in the beginning of round $r$ (and so flaw $f$ in line 5
%satisfies $f\in F_{\sigma_r}$). This property can be easily verified using induction and Definition~\ref{def:causality}.

\myparagraph{$\pi$-stable strategy} Let us fix a total order $\preceq_\pi$ on $F$
defined by some permutation $\pi$ of $F$.
Consider a version of Algorithm~\ref{alg:parallel} where flaw $f$ in line 5 is selected as the lowest flaw
in $F_\sigma-\Gamma^+(I)$ (with respect to $\preceq_\pi$). This corresponds to Algorithm~\ref{alg} with
a specific strategy $\Lambda$; this strategy will be called {\em $\pi$-stable}.
It coincides with the {\tt MaximalSetResample} algorithm of Harvey and Vondr\'ak~\cite{HarveyVondrak15}.
%A very closely related algorithm was analyzed by Kolipaka and Szegedy~\cite{KolipakaSzegedy}
%and Harvey and Vondr\'ak~\cite{HarveyVondrak15} (who assumed additionally
%that $I$ was chosen as the {\em lexicographically first} maximal independent set in $F_\sigma$).

%Although we focus on the commutative case, we will also
%state results for $\pi$-stable strategies since they
%follow automatically from the proof (which is based on the analysis of {\em $\pi$-stable walks}).

\subsection{Algorithmic LLL conditions}\label{sec:LLLconditions}
In this section we formulate sufficient conditions under which a flawless object will be guaranteed to exist. The conditions involve two vectors,
 $\lambda$ and $\mu$. Roughly speaking,
$\lambda$ characterizes sampling oracles and $\mu$ characterizes graph $(F,\sim)$.

\begin{definition} The pair $(\rho,\sim)$ is said to satisfy Algorithmic LLL conditions if
there exist vectors $\lambda,\mu\in\mathbb R^{|F|}$ such that
%Vector $\lambda=\{\lambda_f\:|\:f\in F\}$ will be called a vector of {\em flaw charges} if it satisfies
\begin{subequations}\label{eq:Condition}
\begin{IEEEeqnarray}{cl}
%\lambda_f \ge \sum_{\sigma\in f:\sigma'\in A(f,\sigma)}\rho(\sigma'|f,\sigma)\frac{\omega(\sigma)}{\omega(\sigma')}&\quad\qquad\forall f\in F,\sigma'\in\Omega \label{eq:Condition:a} \\
\lambda_f \ge \sum_{\sigma\in f}\rho(\sigma'|f,\sigma)\frac{\omega(\sigma)}{\omega(\sigma')}&\quad\qquad\forall f\in F,\sigma'\in\Omega \label{eq:Condition:a} \\
\frac{\lambda_f}{\mu_f} \sum_{S\in\IND(f)}\mu(S)\le\theta &\quad\qquad \forall f\in F \label{eq:Condition:b}
\end{IEEEeqnarray}
\end{subequations}
where  $\theta\in(0,1)$ is some constant and $\mu(S)=\prod_{g\in S}\mu_g$. 
\end{definition}
Of course, vector $\lambda$ can be easily eliminated from~\eqref{eq:Condition}.
However, it is convenient to have it explicitly since in many cases it has a natural interpretation.

Note that for resampling oracles~\cite{HarveyVondrak15} one
has $\lambda_f=\omega(f)=\sum_{\sigma\in f}\omega(\sigma)$ and an equality in~\eqref{eq:Condition:a}.
The more general condition~\eqref{eq:Condition:a} was introduced in~\cite{Achlioptas:new,LLL:focs};
%
%This definition brings together ideas from~\cite{HarveyVondrak15,Achlioptas:new,LLL:focs}; 
for more details we refer
to~\cite{LLL:submitted}.

%%%Achlioptas and Iliopoulos~\cite{Achlioptas:new} called $\lambda$ {\em flaw charges},
%%%though instead of~\eqref{eq:Condition:a} they used slightly stronger conditions. 
%%%Namely, they considered the following cases; it is straightforward to check
%%%that in each one of them vector $\lambda$ satisfies~\eqref{eq:Condition:a}: 
%%%\begin{itemize}
%%%\item The case from their earlier work~\cite{Achlioptas}, which in the current terminology can
%%%be described as follows: $\omega$ is a uniform distribution over $\Omega$,
%%%$\rho(\cdot|f,\sigma)$ is a uniform distribution over $A(f,\sigma)$, and $\rho$ is atomic.
%%%They then defined $\lambda_f=1/\min_{\sigma\in f}|A(f,\sigma)|$.

%%%\item Regenerating resampling oracles of Harvey and Vondr\'ak~\cite{HarveyVondrak15}
%%%specified by the equation
%%%\begin{equation*}
%%%\frac{1}{\omega(f)}\sum_{\sigma\in f}\rho(\sigma'|f,\sigma)\omega(\sigma)=\omega(\sigma')
%%%\qquad\forall f\in F,\sigma'\in\Omega
%%%\end{equation*}
%%%where $\omega(f)=\sum_{\sigma\in f}\omega(\sigma)$. In this case Achlioptas and Iliopoulos~\cite{Achlioptas:new} defined $\lambda_f=\omega(f)$.

%%%\item In the general case,~\cite{Achlioptas:new} defined flaw charges via
%%%\begin{equation*}
%%%\lambda_f=b_f\max_{\sigma\in f,\sigma'\in A(f,\sigma)}\left\{\rho(\sigma'|f,\sigma)\frac{\omega(\sigma)}{\omega(\sigma')}\right\} 
%%%\end{equation*}
%%%where $b_f=\max\limits_{\sigma'\in \Omega}|\{\sigma\in f:\sigma'\in A(f,\sigma)\}|$. 
%%%\end{itemize}

\begin{remark}\label{remark:cluster-expansion}
An alternative condition that appeared in the literature (for certain $\lambda$'s) is 
\begin{equation}
\frac{\lambda_f}{\mu_f} \sum_{S\subseteq \Gamma(f)}\mu(S)=\frac{\lambda_f}{\mu_f} \prod_{g\in \Gamma(f)}(1+\mu_g)\le\theta \label{eq:Condition:b'}
\end{equation}
Clearly, \eqref{eq:Condition:b} is weaker than \eqref{eq:Condition:b'}.
We mention that \eqref{eq:Condition:b'} is analogous to the original assymetric LLL condition in~\cite{LLL75},
while~\eqref{eq:Condition:b} corresponds to the {\em cluster expansion} improvement by Bissacot et al.~\cite{Bissacot}
(with the matching algorithmic version by Pedgen~\cite{Pedgen} who considered the variable model  of Moser and Tardos).
It is known that the cluster expansion version can give better results for some applications, see e.g.~\cite{Boettcher,Ndreca,HarveyVondrak15}.

Please note a slight technical difference between the cluster expansion condition in~\cite{Bissacot}
and condition~\eqref{eq:Condition:b}: 
in the former the sum was taken over  $S\in{\tt Ind}(\Gamma^+(f))$,
while in the latter the sum is over $S\in{\tt Ind}(\Gamma(f))$.
The latter condition can be weaker, if $\Gamma^+(f)\ne \Gamma(f)$ for some $f\in F$.
\end{remark}

\myparagraph{Shearer's condition}%\label{def:Shearer}
Shearer~\cite{Shearer} gave a sufficient and necessary condition for a general LLL to hold
for a given dependency graph. Kolipaka and Szegedy~\cite{KolipakaSzegedy}
showed that this condition is sufficient for the Moser-Tardos algorithm,
while Harvey and Vondr\'ak~\cite{HarveyVondrak15} generalized the analysis to 
the framework with resampling oracles. We will show that the same analysis holds for the
framework considered in this paper.
%Using ideas from~\cite{KolipakaSzegedy,HarveyVondrak15},
%we will extend the 
%It is also posssible to formulate an analogue of the Shearer's condition
%for the Random Walk formulation, using ideas from Kolipaka and Szegedy~\cite{KolipakaSzegedy}
%and Harvey and Vondr\'ak~\cite{HarveyVondrak15}.

Consider vector $p\in\mathbb R^{|F|}$.
For a subset $S\subseteq F$ 
denote $p^S=\prod_{f\in S}p_f$; this is a monomial in variables $\{p_f\:|\:f\in F\}$.
Also, define polynomial $q_S$ as follows:
\begin{equation}
q_S=q_S(p)=\sum_{I:S\subseteq I\in\Ind(F)} (-1)^{|I|-|S|}p^I
\end{equation}
%\begin{definition}
\begin{definition}\label{def:Shearer}
Vector $p$ is said to satisfy the Shearer's condition if $q_S(p)\ge 0$ for all $S\subseteq F$,
and $q_\varnothing(p)>0$.

The pair $(\rho,\sim)$ is said to satisfy Shearer's condition if there exist vector $p$ satisfying
Shearer's condition, vector $\lambda$ satisfying~\eqref{eq:Condition:a},
and a constant $\theta\in(0,1)$ such that $\lambda_f\le\theta\cdot p_f$ for all $f\in F$.
\end{definition}
\begin{remark}
Bissacot et al.\ proved in~\cite{Bissacot} that the cluster expansion condition implies Shearer's condition. 
(A more explicit proof can be found in~\cite{HarveyVondrak15}).
However, as mentioned in Remark~\ref{remark:cluster-expansion}, the condition in~\cite{Bissacot}
is slightly different from~\eqref{eq:Condition:b}, so we will not be able to use this implication.
\end{remark}
\begin{remark}
Note that both in condition~\eqref{eq:Condition:b} and in the Shearer's condition we require
constant $\theta$ to be strictly smaller than $1$; thus, we assume that conditions hold with
some slack. In the literature analogous conditions are often formulated
without slack, i.e.\ with $\theta=1$~\cite{Shearer,Bissacot,Pedgen,KolipakaSzegedy}.
The difference between slack and non-slack versions has been recently studied
in detail by Harvey and Vondr\'ak. Interestingly,
if either the cluster expansion condition or the Shearer's condition
holds with $\theta=1$ then it also holds with $\theta<1$,
because conditions define open regions~\cite{HarveyVondrak15}.

We do not know whether the same is true in the case of condition~\eqref{eq:Condition:b}
(since it is slightly different from the cluster expansion condition).
\end{remark}
%(eventhough in many applications the two conditions coincide, since usually $f\in\Gamma(f)$).
%\end{definition}

\section{Our results}\label{sec:results}

%For a flaw $f\in F$ define
%\begin{equation}
%A_f=\min_{\sigma\in f} |A(f,\sigma)|
%\end{equation}
First, we state our results for the sequential version (Algorithm~\ref{alg}).
Unless mentioned otherwise, the flaw selection strategy and the initial distribution $\omegainit$ are assumed to be arbitrary.

\begin{theorem}\label{th:seq}
Suppose that $(\rho,\sim)$ satisfies either condition~\eqref{eq:Condition}
or  Shearer's condition, and one of the following holds:
\begin{itemize}
%\item[(a)] Algorithm~\ref{alg} uses a $\pi$-stable strategy.
\item[(a)] $(\rho,\sim)$ is weakly commutative and atomic.
\item[(b)] $(\rho,\sim)$ is strongly commutative. % and $\omegainit=\omega$.
\end{itemize}
 Define
$$
\gammainit = \max_{\sigma\in\Omega}\frac{\omegainit(\sigma)}{\omega(\sigma)},\qquad
\Indinit = \begin{cases}
\bigcup_{\sigma\in{\tt supp}(\omegainit)}\Ind(F_{\sigma}) & \mbox{in the case (a)} \\
\Ind(F) & \mbox{in the case (b) }
\end{cases}
$$
where ${\tt supp}(\omegainit)=\{\sigma\in\Omega\:|\:\omegainit(\sigma)>0\}$ is the support of $\omegainit$. The probability that Algorithm~\ref{alg} 
produces a flawless object in fewer than $T+r$ steps is at least $1-\theta^r$ where
\begin{eqnarray}\label{eq:Tseq}
T&=&
\frac{\mathlarger 1}{\mathlarger{\log \theta^{-1}}} \left(\log{ \gammainit}
    + \log  \mathlarger{\sum\limits_{R\in\Indinit}}\mu(R)\right) 
\end{eqnarray}
and $\mu(R)=\prod_{f\in R}\mu_f$ (in the case of condition~\eqref{eq:Condition}) or $\mu(R)=\frac{q_R(p)}{q_\varnothing(p)}$
(in the case of the Shearer's condition).
\end{theorem}

%Note that part (a) of Theorem~\ref{th:seq} is a minor variation of existing results~\cite{HarveyVondrak15,Achlioptas:new}
%(except that our precondition~\eqref{eq:Condition:a} unifies conditions in previous works - see Section~\ref{sec:LLLconditions}):
%\begin{itemize}
%\item Harvey and Vondr\'ak~\cite{HarveyVondrak15} proved Theorem~\ref{th:seq}(a)
%in the case of regenerating oracles and distributions $\omegainit=\omega$, with a slightly different expression for $T$.
%\item Achlioptas and Iliopoulos~\cite{Achlioptas:new} proved the result for the ``{\tt RecursiveWalk}'' strategy
%in the special cases described in Section~\ref{sec:LLLconditions} (and assuming condition~\eqref{eq:Condition:b}).
%\end{itemize}
%Parts (b,c) are new results.

%\paragraph{Parallel version}
%The commutativity property can also be used to get a better bound on the number of rounds of the parallel version of the algorithm.
Next, we analyze the parallel version.

\begin{theorem}\label{th:parallel}
Suppose that $(\rho,\sim)$ satisfies either condition~\eqref{eq:Condition}
or the Shearer's condition, and is strongly commutative.
Then the probability that Algorithm~\ref{alg:parallel} 
produces a flawless object in fewer than $T+r$ rounds is at least $1-\theta^r$ where
\begin{eqnarray}\label{eq:Tpar}
T&=&\frac{\mathlarger 1}{{\log \theta^{-1}}} \left( \log{ \gammainit}+\log  \sum\limits_{f\in F}\mu(\{f\})\right) 
\end{eqnarray}
where $\gammainit$ and $\mu(\{f\})$ are as defined in Theorem~\ref{th:seq}.
In particular, $\mu(\{f\})=\mu_f$ in the case of condition~\eqref{eq:Condition} and  $\mu(\{f\})=\frac{q_{\{f\}}(p)}{q_\varnothing(p)}$
in the Shearer's case.
\end{theorem}

\begin{remark}
The possibility of using distribution $\omegainit$ which is different from $\omega$
was first proposed by Achlioptas and Iliopoulos in~\cite{Achlioptas}. Namely, they used a distribution with $|{\tt supp}(\omegainit)|=1$,
and later extended it to arbitrary distributions $\omegainit$ in~\cite{Achlioptas:new}.
There is a trade-off in choosing $\omegainit$: smaller ${\tt supp}(\omegainit)$ leads to a smaller set $\Indinit$
but increases the constant $\gammainit$. It is argued in~\cite{Achlioptas:new}
that using $\omegainit\ne\omega$ can be beneficial when sampling from $\omega$ is a difficult problem,
or when the number of flaws is exponentially large.
\end{remark}

The techniques we develop to deal with the commutative case, yielding Theorems~\ref{th:seq} and~\ref{th:parallel}  above, also give a new result for the non-commutative setting. 
\begin{theorem}\label{th:pistable}
Suppose that Algorithm~\ref{alg} uses a $\pi$-stable strategy and $(\rho,\sim)$ satisfies either condition~\eqref{eq:Condition} or the Shearer's condition. 
Then, the same conclusions hold as for case (a) of Theorem~\ref{th:seq}.
\end{theorem}
We note that while Theorem~\ref{th:seq} was already largely established in earlier works~\cite{HarveyVondrak15, Achlioptas:new, LLL:submitted}, the combination of condition~\eqref{eq:Condition:a}
with the Shearer's condition is new. We include it mainly because its proof serves to introduce several notions from earlier works needed to prove our main results concerning the commutative setting.

\myparagraph{Our techniques} The general idea of the proofs is to construct a ``swapping mapping''
that transforms ``walks'' (which are possible executions of the algorithm) to some canonical
form by applying swap operations from Definition~\ref{def:main}. 
Importantly, we need to make sure that the mapping is injective:
this will guarantee that the sum over original walks is smaller or equal than the sum over ``canonical walks''.
We then upper-bound the latter sum %number of ``canonical walks''
using some standard techniques~\cite{KolipakaSzegedy,HarveyVondrak15}.
We use two approaches:
\begin{enumerate}
\item Theorem~\ref{th:seq}(a): transforming walks to ``forward stable sequences'' (a {\em forward-looking analysis}). This works only in the atomic case (under the
weak commutativity assumption),
and can make use of the knowledge of the set ${\tt supp}(\omegainit)\!$ \linebreak leading to a tighter definition of the set $\Indinit$. % handle the situation when $\omegainit\ne\omega$.
\item Theorems~\ref{th:seq}(b) and~\ref{th:parallel}: transforming walks to ``backward stable sequences'' (a {\em backward-looking analysis}).
This works in the non-atomic cases, but requires strong commutativity.
In this approach the ``roots'' of stable sequences are on the right, and have no connection to $\omegainit$;
this means that we must use $\Indinit=\Ind(F)$. % can only handle the situation when $\omegainit=\omega$.

Analyzing the parallel version requires dealing with ``partial execution logs'' instead of ``full execution logs''.
It appears that this is possible only with backward sequences.
\end{enumerate}

Note that previously a backward-looking analysis (with either ``stable sequences'' or ``witness trees'')
was used for the variable model of Moser and Tardos~\cite{MoserTardos,KolipakaSzegedy,Pedgen}
and for permutations~\cite{Harris:permutations},
while a forward-looking analysis  
was used for LLL versions on non-Cartesian spaces~\cite{Achlioptas,HarveyVondrak15,Achlioptas:new}
and also on Cartesian spaces~\cite{Giotis}.

After the first version of this work~\cite{vnk:LLL:TR}
we learned about a recent book draft by Knuth~\cite{Knuth}. He considers the variable model of Moser-Tardos,
and gives an alternative proof of algorithm's running time 
which is also based on swapping arguments
(justified by a technique of ``coupling'' two random sources, similar to~\cite{MoserTardos}).
We emphasize that we go beyond the variable model, in which case justifying ``swapping'' seems to require different techniques.

The proofs of Theorems~\ref{th:seq} and \ref{th:parallel} are given below in Sections~\ref{sec:PROOF:noncommutative} and~\ref{sec:PROOF:commutative}.
(A brief overview of these proofs can be found in~\cite{LLL:focs}).
The most technical part is probably constructing an injective swapping mapping for transforming to backward stable sequences (proved in Section~\ref{sec:proof:th:ParSwapping}).
In Section~\ref{sec:construction} we describe our third result, which is a proof
of strong commutativity of some existing resampling oracles. We also consider one application,
namely rainbow matchings in complete graphs.

%%%%%%%%%%%%%%%%%%%%%%%%%%%%%%%%%%%%%%%%%%%%%%%%%%%%%%%%%%%%%%%%%%%%%%%%%%%%%%%%%%%%%%%%%%%%%%%%
%%%%%%%%%%%%%%%%%%%%%%%%%%%%%%%%%%%%%%%%%%%%%%%%%%%%%%%%%%%%%%%%%%%%%%%%%%%%%%%%%%%%%%%%%%%%%%%%
%%%%%%%%%%%%%%%%%%%%%%%%%%%%%%%%%%%%%%%%%%%%%%%%%%%%%%%%%%%%%%%%%%%%%%%%%%%%%%%%%%%%%%%%%%%%%%%%
%%%%%%%%%%%%%%%%%%%%%%%%%%%%%%%%%%%%%%%%%%%%%%%%%%%%%%%%%%%%%%%%%%%%%%%%%%%%%%%%%%%%%%%%%%%%%%%%
%%%%%%%%%%%%%%%%%%%%%%%%%%%%%%%%%%%%%%%%%%%%%%%%%%%%%%%%%%%%%%%%%%%%%%%%%%%%%%%%%%%%%%%%%%%%%%%%
%%%%%%%%%%%%%%%%%%%%%%%%%%%%%%%%%%%%%%%%%%%%%%%%%%%%%%%%%%%%%%%%%%%%%%%%%%%%%%%%%%%%%%%%%%%%%%%%

\section{Proof of Theorem~\ref{th:pistable}}\label{sec:PROOF:noncommutative}
%Roughtly speaking, the proof for the commutative case relies on a reduction to the non-commutative
%case. Accordingly, we i

Before jumping to the commutative case, we will need to give a background on the non-commutative
case. Accordingly, in this section we recall a proof of Theorem~\ref{th:pistable}. As mentioned in Section~\ref{sec:results},
various versions of this theorem have been established earlier, % (see~\cite{LLL:submitted} for a discussion),
and techniques used  in this section are a combination
of ideas from previous works~\cite{KolipakaSzegedy,Achlioptas,HarveyVondrak15,Achlioptas:new,LLL:submitted}.
%We begin with some notation and definitions. 
%Let us choose some total order $\preceq_\pi$ on $F$ (specified by a permutation $\pi$ of $F$).

%A sequence of flaws $W=w_1\ldots w_t$ will be called a {\em word}.
We write $f\cong g$ for flaws $f,g\in F$ if either $f\sim g$ or $f=g$
(and $f\not\cong g$ otherwise).

A {\em walk} of length $t$ is a sequence 
$\tau=\sigma_1\RA{w_1}\sigma_2\ldots\sigma_t\RA{w_t}\sigma_{t+1}$
such that $w_i\in F_{\sigma_i}$ and $\sigma_{i+1}\in A(w_i,\sigma_i)$ for $i\in[t]$. 
Its length is denoted as $|\tau|=t$.
For such a walk we define quantity
\begin{equation}
p(\tau)=\omegainit(\sigma_1)\cdot\prod_{i=1}^t \rho(\sigma_{i+1}|w_i,\sigma_i)  \label{eq:ptaur:def}
\end{equation}

Let $\Lambda$ be the strategy for selecting flaws used in Algorithm \ref{alg}. We assume in the analysis
that this strategy is deterministic, i.e.\ the flaw $w_i$ in a walk $\tau=\sigma_1\RA{w_1}\ldots\RA{w_{i-1}}\sigma_{i}\RA{w_i}\ldots$
is uniquely determined by the previous
history \linebreak $\tau_i=\sigma_1\RA{w_1}\ldots\RA{w_{i-1}}\sigma_{i}$. This assumption can be made w.l.o.g.:
if $\Lambda$ is randomized (i.e.\ a distribution over some set of deterministic strategies) then the claim of Theorem~\ref{th:seq} can be obtained by taking the appropriate expectation
over strategies (whose number is finite for a fixed finite $t$). A similar argument applies to Theorem~\ref{th:parallel}.

A walk $\tau$ %=\sigma_1\RA{w_1}\sigma_2\ldots\sigma_t\RA{w_t}\sigma_{t+1}$ 
of length $t$
that can be produced by Algorithm~\ref{alg} with a positive
probability % with $\sigma_1=\sigmainit$ that follows strategy $\Lambda$ 
will be called
a {\em bad $t$-trajectory}. Equivalently, it is a walk that starts at a state $\sigma\in{\tt supp}(\omegainit)$
and follows strategy $\Lambda$.
Note that it goes only through flawed states (except possibly the last state).
Let $\Bad(t)$ be the set of all bad $t$-trajectories. Clearly, for any $\tau\in\Bad(t)$
the probability that the algorithm will produce $\tau$ equals $p(\tau)$, as defined in~\eqref{eq:ptaur:def}.
This gives
\begin{proposition}\label{prop:tauprob}
The probability that Algorithm~\ref{alg} takes $t$ steps or more equals
$\sum_{\tau\in \Bad(t)}p(\tau)$.
\end{proposition}

If $W=w_1\ldots w_t$ is the complete sequence of flaws in a walk $\tau$ then we will write $\tau\myeq W$.
If we want to indicate certain intermediate states of $\tau$ then we will write them in square brackets
at appropriate positions, e.g.\ $\tau\myeq[\sigma_1]w_1w_2[\sigma_3]w_4w_5[\sigma_6]$.

In general, a sequence of flaws will be called a {\em word},
and a sequence of flaws together with some intermediate states (such as $[\sigma_1]w_1w_2[\sigma_3]w_4w_5[\sigma_6]$)
will be called a {\em pattern}.
For a pattern $X$ we define  $\SETOF{X}=\{\tau\:|\:\tau\myeq X\}$ to be the set of walks consistent with $X$. % that start at a state $\sigma\in{\tt supp}(\omegainit)$.
%The length of $X$ (i.e.\ the number of flaws in it) is denoted as $|X|$.
The number of flaws in $X$ is denoted as $|X|$.
The following has been shown in~\cite{LLL:submitted} using an induction argument.

\begin{lemma}[\cite{LLL:submitted}]
For any word $W$ and state $\sigma$ we have
\begin{subequations}\label{eq:WbracketsSum}
\begin{eqnarray}
\sum_{\tau\in\SETOF{W[\sigma]}}p(\tau)&\le&\gammainit\cdot \lambda_W\cdot \omega(\sigma)  \label{eq:WbracketsSum:a} \\
\sum_{\tau\in\SETOF{W}}p(\tau)&\le&\gammainit\cdot \lambda_W  \label{eq:WbracketsSum:b} 
\end{eqnarray}
\end{subequations}
where for a word $W=w_1\ldots w_t$ we denoted
$
\lambda_W = \prod_{i=1}^t \lambda_{w_i}
$. (As described in the previous paragraph, $\SETOF{W[\sigma]}$ is the set of walks $\tau$ whose sequence of flaws is $W$ and the last state is $\sigma$.)
 \label{lemma:WbracketsSum}
\end{lemma}
%\begin{proof}
%Summing~\eqref{eq:WbracketsSum:a} over $\sigma\in\Omega$ gives~\eqref{eq:WbracketsSum:b},
%so it suffices to prove the former inequality. 
%We use induction on the  length of $W$. If $W$ is empty then $\SETOF{W[\sigma]}$ contains a single walk $\tau$ with the state $\sigma$;
%we then have $p(\tau)=\omegainit(\sigma)$, and the claim follows from the definition of $\gammainit$
%in Theorem~\ref{th:seq}. This establishes the base case. Now consider a word $W'=Wf$ with $f\in F$ and a state $\sigma'$.
%We can write
%\begin{eqnarray*}
%\sum_{\tau'\in\SETOF{W'[\sigma']}}p(\tau')
%&\!\!=\!\!&\!\!\!\sum_{\sigma:\sigma'\in A(f,\sigma)}\sum_{\tau\in\SETOF{W[\sigma]}}p(\tau)\cdot \rho(\sigma'|f,\sigma) \\
%&\!\!\stackrel{\mbox{\tiny(a)}}\le\!\!&\!\!\!\sum_{\sigma:\sigma'\in A(f,\sigma)} \gammainit \cdot \lambda_W \cdot \omega(\sigma)    \cdot \rho(\sigma'|f,\sigma) \\
%&\stackrel{\mbox{\tiny(b)}}\le& \gammainit \cdot \lambda_W \cdot [\lambda_f \cdot \omega(\sigma')]
%=\gammainit\cdot \lambda_{W'}\cdot \omega(\sigma')
%\end{eqnarray*}
%where (a) is by the induction hypothesis, and (b) follows from~\eqref{eq:Condition:a}. This gives the induction step, and thus concludes the proof of the lemma.
%\end{proof}

The following technical result will also be useful.
\begin{lemma}\label{lemma:technical}
Consider a walk  
$\tau\myeq \ldots[\sigma]u_1\ldots u_k\ldots$ 
%$\tau=W[\sigma]u_1\ldots u_k W'[\sigma_{t+1}]$ 
where $k\ge 1$. 
Suppose at least one of the following holds: \\
(a) $u_k$ is not present in $\sigma$. \\
(b) $u_1=u_k$. \\
Then there exists index $i\in[k-1]$ such that $u_i\sim u_k$.
\end{lemma}
\begin{proof} We will assume that $\tau=\ldots\sigma_1\RA{u_1}\sigma_{2}\RA{u_{2}}\ldots\RA{u_{k-1}}\sigma_k\RA{u_k}\sigma_{k+1}\ldots$ 
where $\sigma_1=\sigma$.

\noindent{\bf (a)~~} 
Flaw $u_k$ is present in $\sigma_k$ (since $\sigma_k\stackrel{u_k}\rightarrow\sigma_{k+1}$ is a walk)
but not in $\sigma_1$.
Thus, there must exist index $i\in[k-1]$ such that $u_k$ is present in $\sigma_{i+1}$ but not in $\sigma_i$.
We know that $\sigma_i\stackrel{u_i}\rightarrow\sigma_{i+1}$ is a valid walk. Thus, addressing $u_i$ has caused $u_k$ to appear, and therefore $u_i\sim u_k$.

\noindent{\bf (b)~~} Assume that $u_k$ is present in $\sigma_1$ (otherwise the claim holds by (a)).
If $u_k$ is present in $\sigma_{2}$ then $u_k\sim u_k$ (since addressing $u_1=u_k$ at state $\sigma_1$
did not eliminate flaw $u_k$). Otherwise, if $u_k$ is not present in $\sigma_{2}$, we can apply part (a)
and conclude that $u_i\sim u_k$ for some $i\in[2,k-1]$.
\end{proof}

\subsection{Stable walks and stable sequences} \label{sec:stable:definitions}
As shown in~\cite{HarveyVondrak15}, if $\Lambda$ is a $\pi$-stable strategy then
walks $\tau$ that it produces have
a special structure: the word $W$ corresponding to $\tau$ can be uniquely
described by a {\em stable sequence}.
This section gives all necessary definitions.

% (or more precisely words $W$ cor  have a special structure; 
%As we said earlier, mapping $\Phi$ will send walks in $\Bad(t)$
%to {\em $\pi$-stable walks}. To count the number of $\pi$-stable walks, we will use the technique
%from~\cite{KolipakaSzegedy,HarveyVondrak15}. It works by ``compressing'' the walks (in a losless way)
%to a structure called a {\em stable sequence}.
%Such sequences were used   
%We will replace $\pi$-walks with {\em $\pi$-stable walks},
%and then show that $\pi$-stable walks can be compressed into a structure called a {\em stable sequence}
%(instead of a witness forest). We will then use a method for counting stable sequences from~\cite{KolipakaSzegedy,HarveyVondrak15}.
%
%For a subset $S\subseteq F$ denote $\Gamma(S)=\bigcup_{f\in S}\Gamma^+(f)$, where $\Gamma^+(f)=\Gamma(f)\cup\{f\}$.
\begin{definition}\label{def:StableSequence}
A sequence of sets $\varphi=(I_1,\ldots,I_s)$ with $s\ge 1$ is called {\em stable}
if $I_r\in\Ind(F)$ for each $r\in[s]$ and $I_{r+1}\subseteq\Gamma^+(I_r)$ for each $r\in[s-1]$.
\end{definition}
%Note, there is a slight disrepancy between our definition and that of~\cite{KolipakaSzegedy,HarveyVondrak15}:
%the latter had the condition $I_{s+1}\subseteq\Gamma^+(I_s)$ instead of $I_{s+1}\subseteq\Gamma(I_s)$.
\begin{definition}\label{def:Stable}
A word $W=w_1\ldots w_t$ is called {\em stable} if it
 can be partitioned into non-empty words as $W=W_1\ldots W_s$
such that flaws in each word $W_r$ are distinct,
and the sequence $(I_1,\ldots,I_s)$ is stable where $I_r$ is the set of flaws in $W_r$ (for $r\in[s]$).
If in addition each word  
 $W_r=w_i\ldots w_j$ satisfies $w_i\prec_\pi \ldots\prec_\pi w_j$
then $W$ is called {\em $\pi$-stable}.

A walk $\tau\myeq W$ is called stable ($\pi$-stable) if the word $W$ is stable ($\pi$-stable).
\end{definition}

It can be seen that for a stable word the partitioning in Definition~\ref{def:Stable}  is unique, and can be obtained by the following algorithm.
Start with one segment containing $w_1$, and then for $i=2,\ldots,t$ do the following: if there exists flaw $w_k$ in
the last segment with $w_k\cong w_i$ then start a new segment containing $w_i$, otherwise add $w_i$ to the last segment.
(If this algorithm is applied to an arbitrary word $W$ then it may fail to produce a stable sequence
since in the latter case, when $w_i$ is added to the last segment $I_r$, $w_i$ may not belong to $\Gamma^+(I_{r-1})$). 

For a stable word $W$ let $R_W$ be the first set (the ``root'') of the stable sequence $\varphi=(I_1,\ldots,I_s)$
corresponding to $W$,
i.e.\ $R_W=I_1$. (If $W$ is empty then $R_W=\varnothing$). 
Let ${\tt Stab}_\pi$ be the set of $\pi$-stable words $W$ for which there exists a walk $\tau$ with $\tau\myeq W$.
Denote 
\begin{eqnarray*}
{\tt Stab}_\pi(R)&=&\{W\in{\tt Stab}_\pi\::\:R_W=R\}\\
%{\tt Stab}_\pi(t)&=&\{W\in{\tt Stab}_\pi\::\:|W|\ge t\}\\
%{\tt Stab}_\pi(R, t)&=&{\tt Stab}_\pi(R)\cap{\tt Stab}_\pi(t)\\
%
{\tt Stab}_\pi(R,t)&=&\{W\in{\tt Stab}_\pi(R)\::\:|W|\ge t\}
\end{eqnarray*}
The following result is proven in Section~\ref{sec:proof:StabCounting} using techniques from~\cite{KolipakaSzegedy,HarveyVondrak15}.

\begin{theorem}\label{th:StabCounting}
Suppose that $(\rho,\sim)$ satisfies either the cluster expansion condition~\eqref{eq:Condition:b}
or the Shearer's condition from Definition~\ref{def:Shearer}. Then
\begin{equation}\label{eq:StabCounting}
\sum_{W\in{\tt Stab}_\pi(R,t)} \lambda_W \le \mu(R)
\cdot \theta^{t} \qquad \forall R \in \Ind(F)
\end{equation}
Recall that $\mu(R)=\prod_{f\in R}\mu_f$ in the case of condition~\eqref{eq:Condition} and $\mu(R)=\frac{q_R(p)}{q_\varnothing(p)}$
in the Shearer's case.
\end{theorem}

For some parts of the proofs we will need to use the reverse of stable walks and sequences.
(This will correspond to a ``backward-looking analysis'' instead of the ``forward-looking analysis'').
The necessarly definitions are given below.

A sequence $\phi=(I_1,\ldots,I_s)$ will be called {\em reversely stable} if its reverse $(I_s,\ldots,I_1)$ is a stable sequence.
A word $W=w_1\ldots w_t$ will be called {\em reversely stable} ({\em reversely $\pi$-stable}) if its reverse $w_t\ldots w_1$
is stable ($\pi$-stable). A {\em reversely stable} and {\em reversely $\pi$-stable} walks $\tau$ are defined in
an analogous way.

For a reversely stable word $W$ let $R^{\tt rev}_W$ be the last set of the reversely stable sequence $\phi=(I_1,\ldots,I_s)$
corresponding to $W$, i.e.\ $R^{\tt rev}_W=I_s$.
Let ${\tt Stab}^{\tt rev}_\pi$ be the set of reversely $\pi$-stable words $W$ for which there exists a walk $\tau$ with $\tau\myeq W$.
For $R \subseteq F$, denote 
\begin{eqnarray*}
{\tt Stab}^{\tt rev}_\pi(R)&=&\{W\in{\tt Stab}^{\tt rev}_\pi\::\:R^{\tt rev}_W=R\}\\
%{\tt Stab}^{\tt rev}_\pi(t)&=&\{W\in{\tt Stab}^{\tt rev}_\pi\::\:|W|\ge t\}\\
%{\tt Stab}^{\tt rev}_\pi(R, t)&=&{\tt Stab}^{\tt rev}_\pi(R)\cap{\tt Stab}^{\tt rev}_\pi(t)
{\tt Stab}^{\tt rev}_\pi(R,t)&=&\{W\in{\tt Stab}^{\tt rev}_\pi(R)\::\:|W|\ge t\}
\end{eqnarray*}

\begin{theorem}\label{th:StabCounting:rev}
Suppose that $(\rho,\sim)$ satisfies either the cluster expansion condition~\eqref{eq:Condition:b}
or the Shearer's condition from Definition~\ref{def:Shearer}. Then
\begin{equation}\label{eq:StabCounting:rev}
\sum_{W\in{\tt Stab}^{\tt rev}_\pi(R,t)} \lambda_W \le \mu(R) %\frac{q_R(p)}{q_\varnothing(p)}
\cdot \theta^{t} \qquad \forall R \in \Ind(F)
\end{equation}
\end{theorem}

Note that set ${\tt Stab}^{\tt rev}_\pi$ does not necessarily equal the reverse of words from ${\tt Stab}_\pi$
(because of the condition ``there exists walk $\tau$ with $\tau\myeq W$''
present in the definitions of both ${\tt Stab}_\pi$ and ${\tt Stab}^{\tt rev}_\pi$).
Thus, Theorem~\ref{th:StabCounting:rev} does not automatically follow from Theorem~\ref{th:StabCounting}.
Their proofs, however, are very similar (see Section~\ref{sec:proof:StabCounting}).

\subsection{Proof of Theorem~\ref{th:pistable}: a wrap-up}\label{sec:proof:seqa:wrapup}
It is not difficult to show that  a $\pi$-stable strategy produces $\pi$-stable walks (see Proposition~\ref{prop:parallel} below).
The reverse, however, is not necessary true: it
may e.g.\   happen that some flaw $f$ does not appear in a $\pi$-stable walk $\tau$, but is present in all states of $\tau$ and would have been selected by any $\pi$-stable  strategy.

\begin{proposition}\label{prop:parallel} Suppose that strategy $\Lambda$ is implemented as in Algorithm~\ref{alg:parallel}
(with a deterministic choice in line 5). Then any $\tau\in\Bad(t)$ is a stable walk.
If in addition $\Lambda$ is a $\pi$-stable strategy (i.e.\ flaw $f$ in line 5 is chosen as the lowest flaw in $F_\sigma-\Gamma^+(I)$ with respect to
 $\preceq_\pi$)
then any $\tau\in\Bad(t)$ is a $\pi$-stable walk.
\end{proposition}
\begin{proof}
Let $s$ be the number of rounds of Algorithm~\ref{alg:parallel} that produced walk $\tau$,
and 
$I_r\subseteq F$ be the set of flaws addressed in round $r\in[s]$,
or equivalently the set $I$ at the end of round $r$
(with a possible exception for $r=s$: $I_r$ may correspond
to the ``intermediate'' set $I$, depending on where walk $\tau$ was ``cut'').
By this  definition, we have $\tau\myeq W_1\ldots W_s$ 
where $W_r$ is a word containing the flaws in $I_r$ in some order (with $|W_r|=|I_r|$).
We will prove that $(I_1,\ldots,I_s)$ is a stable sequence;
%, and furthermore $I_1\subseteq F_{\sigma_1}$ where $\sigma_1$ is the initial state of $\tau$; 
this will imply the first claim of the proposition.

The independence of each set $I_r$ follows directly from the construction.
%For each $f\in I_1$ we have $f\in F_{\sigma_1}$ (otherwise we would get a contradiction by Lemma~\ref{lemma:technical}(a)).
Consider $r\in[2,s]$, and let $\tau\myeq\ldots W_{r-1}[\sigma] W_r\ldots$.
%where segments $f_1\ldots f_k$ and $g_1\ldots  g_\ell$ correspond to sets $I_{r-1}$ and $I_r$ respectively.
We need to show that for each flaw $f$ present in $W_r$ (i.e.\ $f\in I_r$) we have $f\in\Gamma^+(I_{r-1})$.
Suppose it is not the case.
Lemma~\ref{lemma:technical}(a) gives that $f$ is present in $\sigma$.
%i.e.\ $f\in F_\sigma$. 
Therefore, set $F_\sigma-\Gamma^+(I_{r-1})$
is non-empty (it contains $\sigma$).
But then round $r-1$ would not have terminated at the state $\sigma$ - a contradiction.

%This gives $g_j\in \Gamma^+(I_{r-1})-I_{r-1}\subseteq \Gamma(I_{r-1})$, as desired.

It remains to consider the case when flaw $f$ in line 5 is chosen as the lowest flaw in $F_\sigma-\Gamma^+(I)$ with respect to
 $\preceq_\pi$. Consider round $r$, and let $W_r=w_i\ldots w_j$.
Definition~\ref{def:causality} and an induction argument
show  that during this round set $F_\sigma-\Gamma^+(I)$ in line 5 shrinks
from iteration to iteration. This implies that $w_i\prec_\pi\ldots\prec_\pi w_j$, as desired.
\end{proof}

We also need the following observation.
\begin{lemma}\label{lemma:StableWalkFirstSet}
If $\tau\myeq W$ is a $\pi$-stable walk starting at state $\sigma_1$
then $R_{W}\in\Ind(F_{\sigma_1})$.
\end{lemma}
\begin{proof}
By definition of a stable word, set $R_{W}$ is independent, and corresponds to some prefix $w_1\ldots w_k$ 
of the word $W$.
It remains to show that for each $i\in [k]$ we have $w_i\in F_{\sigma_1}$. 
This follows from Lemma~\ref{lemma:technical}(a) and the condition that 
 $w_j\nsim w_i$ for all  $j\in[i-1]$.
\end{proof}
We now have all ingredients to prove Theorem~\ref{th:pistable}.
Consider Algorithm~\ref{alg} with a $\pi$-stable strategy.
Then each walk $\tau\myeq W$ from $\Bad(t)$ is $\pi$-stable, with $W\in{\tt Stab}_\pi(t)$.
By the definition of $\Bad(t)$ and Lemma~\ref{lemma:StableWalkFirstSet}
we also know that $R_{W}\in \Indinit =\bigcup_{\sigma\in{\tt supp}(\omegainit)}\Ind(F_{\sigma})$.
Therefore,
\begin{eqnarray*}
Pr[\mbox{\#steps}\ge t]
&=& \sum_{\tau\in\Bad(t)} p(\tau)
\;\le\; \sum_{R\in\Indinit}\sum_{W\in {\tt Stab}_\pi(R,t)} \sum_{\tau\in\SETOF{W}} p(\tau) \\
&\stackrel{\mbox{\tiny(a)}}\le& \sum_{R\in\Indinit}\sum_{W\in {\tt Stab}_\pi(R,t)} \gammainit\cdot\lambda_W
\;\stackrel{\mbox{\tiny(b)}}\le \; \gammainit\cdot\sum_{R\in\Indinit} \mu(R)\cdot \theta^t 
\; = \; \theta^{t-T}
\end{eqnarray*}
where (a) holds by Lemma~\ref{lemma:WbracketsSum}, (b) holds by Theorem~\ref{th:StabCounting}, and $T$ is given by the expression in~\eqref{eq:Tseq}:
\begin{eqnarray*}
T&=&
\frac{\mathlarger 1}{\mathlarger{\log \theta^{-1}}} \left(\log{ \gammainit}
    + \log  \mathlarger{\sum\limits_{R\in\Indinit}}\mu(R)\right) 
\end{eqnarray*}

%%%%%%%%%%%%%%%%%%%%%%%%%%%%%%%%%%%%%%%%%%%%%%%%%%%%%%%%%%%%%%%%%%%%%%%%%%%%%%%%%%%%%%%%%%%%%%%%
%%%%%%%%%%%%%%%%%%%%%%%%%%%%%%%%%%%%%%%%%%%%%%%%%%%%%%%%%%%%%%%%%%%%%%%%%%%%%%%%%%%%%%%%%%%%%%%%
%%%%%%%%%%%%%%%%%%%%%%%%%%%%%%%%%%%%%%%%%%%%%%%%%%%%%%%%%%%%%%%%%%%%%%%%%%%%%%%%%%%%%%%%%%%%%%%%
%%%%%%%%%%%%%%%%%%%%%%%%%%%%%%%%%%%%%%%%%%%%%%%%%%%%%%%%%%%%%%%%%%%%%%%%%%%%%%%%%%%%%%%%%%%%%%%%
%%%%%%%%%%%%%%%%%%%%%%%%%%%%%%%%%%%%%%%%%%%%%%%%%%%%%%%%%%%%%%%%%%%%%%%%%%%%%%%%%%%%%%%%%%%%%%%%
%%%%%%%%%%%%%%%%%%%%%%%%%%%%%%%%%%%%%%%%%%%%%%%%%%%%%%%%%%%%%%%%%%%%%%%%%%%%%%%%%%%%%%%%%%%%%%%%

%: Proof of Theorems~\ref{th:seq} and~\ref{th:parallel}

\section{Commutativity and swapping mappings}\label{sec:PROOF:commutative}
From now on we assume that $(\rho,\sim)$ is weakly commutative.
Therefore, for any walk  $\tau\!=\!\shortldots\sigma_1\RA{f}\sigma_2\RA{g}\sigma_3\shortldots$ with $f\not\cong g$
there exists another walk $\tau'\!=\!\shortldots\sigma_1\RA{g}\sigma'_2\RA{f}\sigma_3\shortldots$ obtained
from $\tau$ by applying the ${\tt SWAP}$ operator to the subwalk $\sigma_1\RA{f}\sigma_2\RA{g}\sigma_3$.
Such operation will be called a {\em valid swap} applied to $\tau$.
A mapping $\Phi$ on a set of walks that works by applying some sequence of valid swaps will be called a {\em swapping mapping}.
Note that if $\tau'=\Phi(\tau)$ then the first and the last states of $\tau'$ coincide with that of $\tau$,
and $\lambda_{W'}=\lambda_W$ where $\tau\myeq W$, $\tau'\myeq W'$.
Furthermore, if $(\rho,\sim)$ is strongly commutative then
 $p(\tau')=p(\tau)$.

We now deal with  the case when $\Lambda$ is an arbitrary deterministic strategy,
and so walks $\tau\in\Bad(t)$ are not necessarily $\pi$-stable. Our approach will be
to construct a {\em bijective} swapping mapping $\Phi$ that sends walks $\tau\in\Bad(t)$
to some canonical walks, namely either to $\pi$-stable walks (which will work only in the atomic case)
or to reversely $\pi$-stable walks (which will work in the general case).
%Theorems~\ref{th:Lambda}, \ref{th:ParSwapping} 
Theorems~\ref{th:Swapping}(a), \ref{th:Swapping}(b)
and \ref{th:ParSwapping:f} below give three ways to construct such mappings;
they will be used for the proofs of Theorems~\ref{th:seq}(a), \ref{th:seq}(b) and \ref{th:parallel}, respectively.

We will need a few definitions first. 
A {\em generalized walk} is a formal finite sequence $\tau=[\sigma_1]w_1[\sigma_2]w_2\ldots $
with $w_i\in F_{\sigma_i}$ and $\sigma_{i+1}\in A(w_i,\sigma_i)$ for all $i$.
Note that $\tau$ can either end with a state ($\tau=\ldots [\sigma_{t}]$),
or end with a flaw ($\tau=\ldots [\sigma_t] w_t$), or be empty ($\tau=\epsilon$).
In the first case $\tau$ is a usual walk. To indicate this case, we will write $\tau=\ldots[\Omega]$.
We emphasize that by a ``walk'' we always mean a sequence of the form $\tau=\ldots[\Omega]$,
unless we explicitly use the word ``generalized''.
For two generalized walks $\tau,\tau'$ their largest common prefix is denoted as $\tau\wedge \tau'$ 
(it is itself a generalized walk).
\begin{definition}
A set of walks $\calX$ is called {\em valid} if $\tau\wedge \tau'\ne\ldots [\Omega]$ for all distinct $\tau,\tau'\in\calX$.
\end{definition}
It can be seen that $\calX$ is valid if two conditions hold:
(i) there exists a deterministic strategy $\tilde\Lambda$ in Algorithm~\ref{alg}
such that all walks in $\calX$ follow $\tilde\Lambda$, and
(ii) for any $\tau,\tau'\in \calX$, walk $\tau$ is not a proper prefix of $\tau'$.
In particular, set $\Bad(t)$ is valid.

%Walk $\tau$ will be called a {\em  prefix} of a walk $\tau'$ if $\tau'$ starts with $\tau$.
%Walk $\tau$ is a {\em proper prefix} of $\tau'$ if in addition $\tau'\ne \tau$.
%A word $W$ is called a prefix of $\tau$ if $\tau\myeq WU$ for some word $U$.

%A set of walks $\calX$ will be called {\em valid} if (i) all walks in $\calX$ follow
%the same deterministic strategy (not necessarily the one used in Algorithm~\ref{alg}), and (ii) for any $\tau,\tau'\in\calX$
%the walk $\tau$ is not a proper prefix of $\tau'$.

\begin{theorem}\label{th:Swapping}  Suppose that $(\rho,\sim)$ is  weakly commutative, and $\calX$
is a valid set of walks. \\
(a) Suppose in addition that $(\rho,\sim)$ is atomic. Then there exists a set of $\pi$-stable walks~$\calX_\pi$ and a swapping mapping $\Phi:\calX\rightarrow \calX_\pi$ which is a bijection. \\
(b) In a general case there exists a set of  reversely $\pi$-stable walks $\calX_\pi$ and a swapping mapping $\Phi:\calX\rightarrow \calX_\pi$ which is a bijection.
%and (ii) if $\tau_1,\tau_2\in\Bad(t)$ are distinct $t$-trajectories then $\Phi(\tau_1)\ne \Phi(\tau_2)$.
\end{theorem}

%\begin{theorem}\label{th:ParSwapping}
%Suppose that $(\rho,\sim)$ weakly commutative, and $\calX$ is a valid set of walks.
%There exists a set of  reversely $\pi$-stable walks $\calX_\pi$ and a swapping mapping $\Phi:\calX\rightarrow \calX_\pi$ which is a bijection.
%%such that for any word $W$  the set $\{\tau\in \calX_\pi\:|\:\tau\myeq W\}$ is valid.
%\end{theorem}
A word $W$ is called a prefix of $\tau$ if $\tau\myeq WU$ for some word $U$.
For a walk $\tau$ containing flaw $f$ we define word ${\tt PREFIX}^f(\tau)$ as the longest prefix of $\tau$ that ends with $f$.
Thus, we have $\tau\myeq {\tt PREFIX}^f(\tau) U$ where
 ${\tt PREFIX}^f(\tau)=\ldots f$ and  word $U$ does not contain $f$.

\begin{theorem}\label{th:ParSwapping:f}
Suppose that $(\rho,\sim)$ weakly commutative, and $\calX$ is a valid set of walks containing some fixed flaw $f\in F$.
%Fix $f\in F$, and let $\calX$ be a valid set of walks containing $f$. % (if $f\in F$).
There exists a set of walks~$\calX_\pi$ and a swapping mapping $\Phi:\calX\rightarrow \calX_\pi$ which is a bijection
such that \\
(i) for any $\tau\in \calX_\pi$ the word $W={\tt PREFIX}^f(\tau)$ is reversely $\pi$-stable with $R^{\tt rev}_W=\{f\}$; \\
(ii) for any word $W$ the set  $\mbox{$\{\tau\in \calX_\pi\:|\:{\tt PREFIX}^f(\tau)=W\}$}$ is valid.
\end{theorem}

We prove Theorems~\ref{th:Swapping}-\ref{th:ParSwapping:f} in Sections~\ref{sec:proof:th:Lambda} and~\ref{sec:proof:th:ParSwapping},
but first we show how they imply Theorems~\ref{th:seq} and~\ref{th:parallel}.

%%%%%%%%%%%%%%%%%%%%%%%%%%%%%%%%%%%%%%%%%%%%%%%%%%%%%%%%%%%%%%%%%%%%%%%%%%%%%%%%%%%%%%%%%%%%%%%%%%%%%%%%%%%%%%%%%%%%%%%%%%%%%%
%%%%%%%%%%%%%%%%%%%%%%%%%%%%%%%%%%%%%%%%%%%%%%%%%%%%%%%%%%%%%%%%%%%%%%%%%%%%%%%%%%%%%%%%%%%%%%%%%%%%%%%%%%%%%%%%%%%%%%%%%%%%%%
%%%%%%%%%%%%%%%%%%%%%%%%%%%%%%%%%%%%%%%%%%%%%%%%%%%%%%%%%%%%%%%%%%%%%%%%%%%%%%%%%%%%%%%%%%%%%%%%%%%%%%%%%%%%%%%%%%%%%%%%%%%%%%
%%%%%%%%%%%%%%%%%%%%%%%%%%%%%%%%%%%%%%%%%%%%%%%%%%%%%%%%%%%%%%%%%%%%%%%%%%%%%%%%%%%%%%%%%%%%%%%%%%%%%%%%%%%%%%%%%%%%%%%%%%%%%%

\subsection{Proof of Theorem~\ref{th:seq}(a) (sequential algorithm in the atomic case)}\label{sec:proof:atomic}
The assumption that  $(\rho,\sim)$ is atomic gives the following observation.
\begin{proposition}[\cite{Achlioptas}]\label{prop:AtomicReconstruction}
Walk $\tau=\sigma_1\RA{w_1}\sigma_2\ldots\sigma_t\RA{w_t}\sigma_{t+1}$ can be uniquely reconstructed from the sequence of flaws $w_1\ldots w_t$ and the final state $\sigma_{t+1}$.
\end{proposition}
\begin{proof}
By atomicity, state $\sigma_i$ can be uniquely reconstructed from the flaw $w_i$ and the state $\sigma_{i+1}$.
Applying this argument for $i=t,t-1,\ldots,1$ gives the claim.
\end{proof}
%The proposition allows us to write $\tau= w_1\ldots w_t[\sigma_{t+1}]$ instead of $\tau\myeq w_1\ldots w_t[\sigma_{t+1}]$.
The proposition allows us to write walks more compactly as $\tau=w_1\ldots w_t[\sigma_{t+1}]$.
Also, Lemma~\ref{lemma:WbracketsSum} gives for a walk $\tau=W[\sigma_{t+1}]$ that
\begin{eqnarray}
p(\tau)&\le& \gammainit \cdot \lambda_W \cdot \omega(\sigma_{t+1}) \label{eq:WbracketsSum:atomic}
\end{eqnarray}

Let us apply Theorem~\ref{th:Swapping}(a) with $\calX=\Bad(t)$. Denoting $\Bad_\pi(t)=\calX_\pi$, we can write
%We can now prove Theorem~\ref{th:seq}(a):
\begin{eqnarray*}
Pr[\mbox{\#steps}\ge t]
&=& \sum_{\tau\in\Bad(t)} p(\tau)
\;\stackrel{\mbox{\tiny(a)}}\le\; \sum_{\tau=W[\sigma]\in\Bad(t)} \gammainit\cdot \lambda_W \cdot \omega(\sigma) \\
&\stackrel{\mbox{\tiny(b)}}=& \sum_{\tau=W[\sigma]\in\Bad_\pi(t)} \gammainit\cdot \lambda_W \cdot \omega(\sigma) \\
&\le& \gammainit\cdot\sum_{R\in\Indinit}\sum_{W\in {\tt Stab}_\pi(R,t)} \sum_{\sigma\in\Omega}  \lambda_W \cdot \omega(\sigma) \\
&\le& \gammainit\cdot\sum_{R\in\Indinit} \sum_{W\in {\tt Stab}_\pi(R,t)}   \lambda_W  
\;\le \; \gammainit\cdot\sum_{R\in\Indinit} \mu(R)\cdot \theta^t 
\; = \; \theta^{t-T}
\end{eqnarray*}
where in (a) we used eq.~\eqref{eq:WbracketsSum:atomic}, in (b) we use bijectiveness of mapping $\Phi$, %follows from Theorem~\ref{th:Lambda},
and the rest is similar to the derivation in Section~\ref{sec:proof:seqa:wrapup}.

%%%%%%%%%%%%%%%%%%%%%%%%%%%%%%%%%%%%%%%%%%%%%%%%%%%%%%%%%%%%%%%%%%%%%%%%%%%%%%%%%%%%%%%%%%%%%%%%%%%%%%%%%%%%%%%%%%%%%%%%%%%
%%%%%%%%%%%%%%%%%%%%%%%%%%%%%%%%%%%%%%%%%%%%%%%%%%%%%%%%%%%%%%%%%%%%%%%%%%%%%%%%%%%%%%%%%%%%%%%%%%%%%%%%%%%%%%%%%%%%%%%%%%%
%%%%%%%%%%%%%%%%%%%%%%%%%%%%%%%%%%%%%%%%%%%%%%%%%%%%%%%%%%%%%%%%%%%%%%%%%%%%%%%%%%%%%%%%%%%%%%%%%%%%%%%%%%%%%%%%%%%%%%%%%%%
%%%%%%%%%%%%%%%%%%%%%%%%%%%%%%%%%%%%%%%%%%%%%%%%%%%%%%%%%%%%%%%%%%%%%%%%%%%%%%%%%%%%%%%%%%%%%%%%%%%%%%%%%%%%%%%%%%%%%%%%%%%
%%%%%%%%%%%%%%%%%%%%%%%%%%%%%%%%%%%%%%%%%%%%%%%%%%%%%%%%%%%%%%%%%%%%%%%%%%%%%%%%%%%%%%%%%%%%%%%%%%%%%%%%%%%%%%%%%%%%%%%%%%%
%%%%%%%%%%%%%%%%%%%%%%%%%%%%%%%%%%%%%%%%%%%%%%%%%%%%%%%%%%%%%%%%%%%%%%%%%%%%%%%%%%%%%%%%%%%%%%%%%%%%%%%%%%%%%%%%%%%%%%%%%%%
%%%%%%%%%%%%%%%%%%%%%%%%%%%%%%%%%%%%%%%%%%%%%%%%%%%%%%%%%%%%%%%%%%%%%%%%%%%%%%%%%%%%%%%%%%%%%%%%%%%%%%%%%%%%%%%%%%%%%%%%%%%

%For the remaining results, namely Theorem~\ref{th:seq}(b) and Theorem~\ref{th:parallel}, we will use a backward-looking analysis, which involves reversely stables sequences.

\subsection{Proof of Theorem~\ref{th:seq}(b) (sequential algorithm with strong commutativity)}

%We now proceed with the proof of Theorem~\ref{th:seq}(b).
In this case we have %$\omegainit=\omega$ (implying $\gammainit=1$) and
 $\Indinit=\Ind(F)$.
Let us apply Theorem~\ref{th:Swapping}(b) with $\calX=\Bad(t)$. Denoting $\Bad_\pi(t)=\calX_\pi$, we can write
\begin{eqnarray*}
Pr[\mbox{\#steps}\ge t]
&=& \sum_{\tau\in\Bad(t)} p(\tau)
\;\stackrel{\mbox{\tiny(a)}}= \sum_{\tau\in\Bad_\pi(t)} p(\tau)  
\;\le\; \sum_{R\in\Ind(F)}\sum_{W\in {\tt Stab}^{\tt rev}_\pi(R,t)} \sum_{\tau\in\SETOF{W}} p(\tau) \\
&\le& \sum_{R\in\Ind(F)}\sum_{W\in {\tt Stab}^{\tt rev}_\pi(R,t)} \gammainit\cdot\lambda_W
\;\le \; \gammainit\cdot\sum_{R\in\Ind(F)} \mu(R)\cdot \theta^t 
\; = \; \theta^{t-T}
%\; \le \; \theta^{t-T}
\end{eqnarray*}
where in (a) we used bijectiveness of mapping $\Phi$
and strong commutativity of $(\rho,\sim)$, and the rest is similar to the derivation in Section~\ref{sec:proof:seqa:wrapup}.

\subsection{Proof of Theorem~\ref{th:parallel} (parallel version with strong commutativity)}
We now  analyze Algorithm~\ref{alg:parallel}
with a deterministic choice of flaw $f$ in line 5.
Equivalently, this can be viewed as running Algorithm~\ref{alg} with
some deterministic strategy. 
We will need the following result.

%We have $\omegainit=\omega$, and thus $\gammainit=1$. %We assume that $(\rho,\sim)$ is strongly commutative.
%, which we denote as $\Lambda$.
%We also assume that $\sigmainit\in\Omega$ was sampled uniformly at random.

%
\begin{lemma}\label{lemma:calX}
Consider a word $W$ and a valid set of walks $\calX$ such that
$W$ is a prefix of every walk in $\calX$.
Then
\begin{equation}\label{eq:calX}
\sum_{\tau\in\calX} p(\tau) \le  \gammainit\cdot\lambda_W 
\end{equation}
\end{lemma}
\begin{proof}
We use induction on $\sum_{\tau\in\calX}(|\tau|-|W|)$. The base case $\sum_{\tau\in\calX}(|\tau|-|W|)=0$ 
is straightforward:
we then have $\calX\subseteq\SETOF{W}$, and so the claim follows from Lemma~\ref{lemma:WbracketsSum}.
Consider a valid set $\calX$ with $\sum_{\tau\in\calX}(|\tau|-|W|)\ge 1$.
Let $\hat\tau$ be a longest walk in $\calX$, then $|\hat\tau|\ge |W|+1$.
Let $\hat\tau^-$ be the proper prefix of $\hat\tau$ of length $|\hat\tau|-1$. 
We have $\hat\tau^-\notin \calX$ since $\calX$ is a valid set.
Define set $\calY$ as follows: $\calY=\{\tau\in\calX\:|\:\mbox{$\hat\tau^-$ is a proper prefix of $\tau$}\}$.
By the choice of $\hat\tau$ we get $|\tau|=|\hat\tau|$ for all $\tau\in\calY$,
and so we must have $\tau=\hat\tau^-\RA{w}\sigma$ for some $w\in F$ and $\sigma\in\Omega$.
Since set $\calY\subseteq\calX$ is valid, the flaw $w$
in the expression $\tau=\hat\tau^-\RA{w}\sigma$ must be the same for all $\tau\in\calY$.
Thus, $\calY=\{\hat\tau^-\RA{w}\sigma\:|\:\sigma\in Y\}$ for some set of flaws $Y\subseteq F$.
In fact, we must have $Y\subseteq A(w,\hat\sigma)$ where $\hat\sigma$ is the final state of~$\hat\tau^-$.

Define $\calX^-=(\calX-\calY)\cup\{\hat\tau^-\}$. We have
$$
  \sum_{\tau\in\calX} p(\tau) - \sum_{\tau\in\calX^-} p(\tau) 
= \left(\sum_{\tau\in\calY} p(\tau)\right) - p(\hat\tau^-)
= p(\hat\tau^-) \cdot \left[ \left(\sum_{\sigma\in Y} \rho(\sigma|w,\hat \sigma)\right) -1\right] 
\le 0
$$
It is straightforward to check that set $\calX^-$ is valid, and $W$ is a prefix of every walk in $\calX^-$.
Using the induction hypothesis for $\calX^-$ and the inequality above gives the claim for $\calX$.
\end{proof}

Consider executions of Algorithm~\ref{alg:parallel} consisting of at least $s$ rounds.
For each such execution let $\tau$ be the walk containing flaws addressed in the first $s-1$ rounds
and the first flaw addressed in round~$s$.
Let $\BadPar(s)$ be the set of such walks $\tau$.
%Clearly, we have $|\tau|\ge s$ for any $\tau\in \BadPar(s)$. 

We say that a word $U=u_1\ldots u_s$ is a {\em chain} of a walk $\tau\myeq w_1\ldots w_t$
if $U$ is a subsequence of the sequence $w_1\ldots w_t$
and $u_i\cong u_{i+1}$ for $i\in[s-1]$.
\begin{proposition}\label{prop:NASJFJSANFA}
(a) For each $\tau\in\BadPar(s)$ the length of a longest chain in $\tau$ equals $s$.
(b) Set $\BadPar(s)$ is valid. 
\end{proposition}
\begin{proof}
Let $I_r$ be the set of flaws addressed in round $r$ (for $r\in[s-1]$),
and let $I_s=\{f\}$ where $f$ is the first flaw addressed in round $s$.
In the proof of Proposition~\ref{prop:parallel} we showed that
$(I_1,\ldots,I_s)$ is a stable sequence
and
$\tau\myeq W_1\ldots W_s$ 
where $W_r$ is a word containing the flaws in $I_r$ in some order (with $|W_r|=|I_r|$).
These facts imply part (a).
Let us prove (b). By construction, all walks in $\BadPar(s)$ follow the same deterministic strategy used in Algorithm~\ref{alg:parallel}.
Now consider a walk $\tau'\in\BadPar(s)$.
By the definition of $\tau'\in\BadPar(s)$, any proper prefix $\tau$ of $\tau'$
corresponds to an execution of Algorithm~\ref{alg:parallel} with 
at most $r-1$ rounds, and so the length of a longest chain in $\tau$ is at most $r-1$.
Thus, $\tau\notin\BadPar(s)$.
\end{proof}

Fix a flaw $f\in F$, and let $\BadPar^f(s)$ be the set of walks $\tau\in\BadPar(s)$ that contain a chain of length $s$
that ends with $f$. 
We now apply Theorem~\ref{th:ParSwapping:f} with the set $\calX=\BadPar^f(s)$.
The set $\calX_\pi$ constructed in the theorem will be denoted as $\BadPar_\pi^f(s)$.
 %, and denote $\BadPar^f_\pi(s)=\calX^f_\pi$.
Since every walk $\tau\in \BadPar^f(s)$ has a chain of length $s$ that ends with $f$, so does every walk $\tau\in \BadPar^f_\pi(s)$
(applying a swapping mapping to $\tau$ does not affect the chain).
%Therefore, the length of the stable sequence for ${\tt REV}[W^f_\tau]$ corresponding to each $\tau\in \calX^f_\pi$ is at least $s$.
This means that $|{\tt PREFIX}^f(\tau)|\ge s$ for each $\tau\in \BadPar^f_\pi(s)$.

For a word $W\in{\tt Stab}^{\tt rev}_\pi(\{f\},s)$ let $\SETOF{W|\neg f}$ be the set of walks $\tau\in \BadPar^f_\pi(s)$
of the form $\tau\myeq WU$ where $U$ does not contain $f$.
Note that $W$ ends with $f$, and thus ${\tt PREFIX}^f(\tau)=W$ for any $\tau\in \SETOF{W|\neg f}$.
Theorem~\ref{th:ParSwapping:f}(i) gives that  $\BadPar^f_\pi(s)=\bigcup_{W\in{\tt Stab}^{\tt rev}_\pi(\{f\},s)}\SETOF{W|\neg f}$.
Also, from Theorem~\ref{th:ParSwapping:f}(ii) we conclude that  $\SETOF{W|\neg f}$ is a valid set for any $W\in{\tt Stab}^{\tt rev}_\pi(\{f\},s)$.
We can thus write
\begin{eqnarray*}
\sum_{\tau\in\BadPar^f(s)} p(\tau)
&\stackrel{\mbox{\tiny(a)}}=&  \sum_{\tau\in\BadPar^f_\pi(s)} p(\tau)
\;=\; \sum_{W\in{\tt Stab}^{\tt rev}_\pi(\{f\},s)} \sum_{\tau\in\SETOF{W|\neg f}} p(\tau) \\
&\stackrel{\mbox{\tiny(b)}}\le& \sum_{W\in{\tt Stab}^{\tt rev}_\pi(\{f\},s)} \gammainit\cdot\lambda_W
\;\stackrel{\mbox{\tiny(c)}}\le\; \gammainit\cdot\mu(\{f\})\cdot \theta^s
\end{eqnarray*}
where  in (a) we used bijectiveness of $\Phi$ and strong commutativity of $(\rho,\sim)$, in (b) we used Lemma~\ref{lemma:calX}, and in (c) we used Theorem~\ref{th:StabCounting:rev}.
By Proposition~\ref{prop:NASJFJSANFA}, $\BadPar(s)=\bigcup_{f\in F}\BadPar^f(s)$, therefore
\begin{eqnarray*}
Pr[\mbox{\#rounds}\ge s]
&\le &\sum_{\tau\in\BadPar(s)} p(\tau) \\
&\le &  \sum_{f\in F}\sum_{\tau\in\BadPar^f(s)} p(\tau) 
\;\le\; \sum_{f\in F} \gammainit\cdot\mu(\{f\})\cdot \theta^s
\;=\; \theta^{s-T}
\end{eqnarray*}
where $T$ is given by the expression in~\eqref{eq:Tpar}.

%%%%%%%%%%%%%%%%%%%%%%%%%%%%%%%%%%%%%%%%%%%%%%%%%%%%%%%%%%%%%%%%%%%%%%%%%%%%%%%%%%%%%%%%%%%%%%%%%%%%%%%%
%%%%%%%%%%%%%%%%%%%%%%%%%%%%%%%%%%%%%%%%%%%%%%%%%%%%%%%%%%%%%%%%%%%%%%%%%%%%%%%%%%%%%%%%%%%%%%%%%%%%%%%%
%%%%%%%%%%%%%%%%%%%%%%%%%%%%%%%%%%%%%%%%%%%%%%%%%%%%%%%%%%%%%%%%%%%%%%%%%%%%%%%%%%%%%%%%%%%%%%%%%%%%%%%%
%%%%%%%%%%%%%%%%%%%%%%%%%%%%%%%%%%%%%%%%%%%%%%%%%%%%%%%%%%%%%%%%%%%%%%%%%%%%%%%%%%%%%%%%%%%%%%%%%%%%%%%%

\subsection{Proof of Theorem~\ref{th:Swapping}(a) (swapping mapping for the atomic case)}\label{sec:proof:th:Lambda}
Recall that in the atomic case the walks can be denoted as $\tau=W[\sigma]$,
since $\tau$ can be uniquely reconstructed from the word $W$ and the state $\sigma$
(see Section~\ref{sec:proof:atomic}).

We can extend a ``valid swap'' operation to words in a natural way:
it is a transformation of the form $\ldots fg \ldots\mapsto \ldots gf \ldots$ where $f\not\cong g$.
We will write $W\equiv W'$ if $W'$ can be obtained from $W$ via a sequence of valid swaps.
Clearly, ``$\equiv$'' is an equivalence relation.
It can  be seen that a walk $\tau=W[\sigma]$ can be transformed to a walk $\tau'=W'[\sigma]$
if and only if $W\equiv W'$. In this case we will write $\tau'\equiv\tau$;
again, ``$\equiv$'' is an equivalence relation on the set of walks.

\begin{lemma}%\label{lemma:GASIFJAOSHGAISF}
For any word $W$ there exists a sequence of valid swaps that transforms $W$ to a $\pi$-stable word.
\end{lemma}
\begin{proof}
It suffices to show that $W$ can be transformed to a stable word via valid swaps.
(Transforming a stable word $W$ to a $\pi$-stable word is straightforward: if
$W=W_1\ldots W_s$ is the partition described in Definition~\ref{def:Stable}
then we simply
need to apply swaps inside each word $W_r$ to ``sort'' it according to $\preceq_\pi$; any such swap will
be valid by the property of $W_r$).

Let $W=\hat{ w}_1\ldots\hat{ w}_t$. We will prove by
induction on $i=1,\ldots,t$ that $W$ can be transformed via valid swaps to a word $W'={ w}_1\ldots{ w}_i\ldots{ w}_t$ such that the prefix ${ w}_1\ldots{ w}_i$ is a stable word.
The base case $i=1$ is trivial. Suppose the claim holds for $i-1$, and let us show it for $i\in[2,t]$.

By the induction hypothesis, $W$ can be transformed to a word \linebreak $W'={ w}_1\ldots{ w}_{i-1}{ f}\ldots{ w}_t$
such that  ${ w}_1\ldots{ w}_{i-1}$ is a stable word.
Let ${ w}_1\ldots{ w}_{i-1}=W_1\ldots W_s$
be the corresponding partition described in Definition~\ref{def:Stable}.
Let $W_r$ be the rightmost word that contains a  flaw ${ g}$ with $f\cong g$.
(If such word doesn't exist then we set $r=0$; thus, $r\in[0,s]$).
If $r=s$ then we can leave the word $W'$ as it is - it satisfies the induction hypothesis for $i$.

Suppose that $r\in[0,s-1]$. Then we repeatedly swap ${ f}$ with the left neighbor, 
stopping when ${ f}$ gets between words $W_r$ and $W_{r+1}$ (or at the first position, if $r=0$).
By the definition of $r$, these are valid swaps. The new word now satisfies the induction hypothesis for $i$
(flaw ${ f}$ will be assigned to word $W_{r+1}$). 
\end{proof}
The lemma means that there exists a swapping mapping $\Phi$ that sends any walk $\tau=W[\sigma]$
to a $\pi$-stable walk $\tau'=W'[\sigma]$. Next, we will show that any such mapping is injective on $\calX$.

\begin{lemma}
Suppose that $\{\tau,\tau'\}$ is a valid set of walks and $\tau\equiv \tau'$. 
Then $\tau=\tau'$.
\end{lemma}
\begin{proof}
We use induction on the length of $\tau$. The base case is trivial: if $|\tau|=|\tau'|=0$
then condition $\tau\equiv \tau'$ implies that $\tau=\tau'=[\sigma]$ for some state $\sigma$.

Suppose that $\tau=fW[\sigma]$.
Condition $\tau\equiv\tau'$ means that $\tau'$ and $\tau$ start at the same state.
Therefore, since $\{\tau,\tau'\}$ is a valid set, the first flaw addressed in $\tau'$ is
the same as in $\tau$. Thus, $\tau'=f W'[\sigma]$.
%Define walks $\tau^-=w_2\ldots w_t[\sigma_{t+1}]$ and $\bar\tau^-=\bar w^-_2\ldots \bar w^-_t[\sigma_{t+1}]$.
The set $\{W[\sigma],W'[\sigma]\}$ must be valid (since the set $\{fW[\sigma],fW'[\sigma]\}$ is valid).
%Clearly, walks $W[\sigma]$  and $W'[\sigma]$ follow the same deterministic strategy.
We will show below that $W[\sigma]\equiv W'[\sigma]$; this will mean that $W[\sigma]=W'[\sigma]$
by the induction hypothesis, thus giving the desired result.

We need to show that $W\equiv W'$. We know that $fW\equiv fW'$,
therefore there exists a sequence of words $U^{(1)},\ldots,U^{(k)}$ with $U^{(1)}=fW$, $U^{(k)}=fW'$
such that $U^{(i+1)}$ is obtained from $U^{(i)}$ via a single valid swap operation.
Let $V^{(i)}$ be the word obtained from $U^{(i)}$ by moving the first occurence of flaw $f$  in $U^{(i)}$
to the first position. Thus, for each $i\in[k]$ we have $V^{(i)}=fW^{(k)}$ for some word $W^{(k)}$.
It can be seen that either $V^{(i+1)}=V^{(i)}$ (if the valid swap
operation applied to $U^{(i)}$ involved the first occurence of $f$)
or $V^{(i+1)}$ is obtained from $V^{(i)}$ via a valid swap  (otherwise).
Thus, either $W^{(i+1)}=W^{(i)}$ or $W^{(i+1)}$ is obtained from $W^{(i)}$ via a valid swap.
It remains to notice that $W^{(1)}=W$ and $W^{(k)}=W'$.
\end{proof}

We can now prove that $\Phi$ is injective on $\calX$.
Suppose that $\Phi(\tau)=\Phi(\tau')$ for walks $\tau,\tau'\in\calX$.
We have $\tau\equiv \Phi(\tau)=\Phi(\tau')\equiv \tau'$, and so $\tau\equiv \tau'$.
From the previous lemma we obtain that $\tau=\tau'$, which concludes the proof of Theorem~\ref{th:Swapping}(a).

%%%%%%%%%%%%%%%%%%%%%%%%%%%%%%%%%%%%%%%%%%%%%%%%%%%%%%%%%%%%%%%%%%%%%%%%%%%%%%%%%%%%%%%%%%%%%%%%%%%%%%%%
%%%%%%%%%%%%%%%%%%%%%%%%%%%%%%%%%%%%%%%%%%%%%%%%%%%%%%%%%%%%%%%%%%%%%%%%%%%%%%%%%%%%%%%%%%%%%%%%%%%%%%%%
%%%%%%%%%%%%%%%%%%%%%%%%%%%%%%%%%%%%%%%%%%%%%%%%%%%%%%%%%%%%%%%%%%%%%%%%%%%%%%%%%%%%%%%%%%%%%%%%%%%%%%%%
%%%%%%%%%%%%%%%%%%%%%%%%%%%%%%%%%%%%%%%%%%%%%%%%%%%%%%%%%%%%%%%%%%%%%%%%%%%%%%%%%%%%%%%%%%%%%%%%%%%%%%%%

\def\G{{\bf G}}
\def\V{{\bf V}}
\def\E{{\bf E}}
\def\W{{\hat{\bf V}}}
\def\I{{\bf I}}

\def\PA{{$(\ast)$}}
\def\PB{{$(\ast\ast)$}}
\def\PC{{$(\ast\!\ast\!\ast)$}}

\subsection{Proof of Theorems~\ref{th:Swapping}(b) and~\ref{th:ParSwapping:f} (swapping mappings for the non-atomic case)}\label{sec:proof:th:ParSwapping}

Consider a word $W=w_1\ldots w_t$. It will be convenient to alternatively write it
as $W={\bf w}_1\ldots {\bf w}_t$
where ${\bf w}_i=(w_i,n_i)$
and $n_i$ counts from the left which occurrence of the flaw $w_i$ it is: $n_i=|\{j\in[i]\:|\:w_j=w_i\}|\ge 1$.
Note that all elements ${\bf w}_1,\ldots,{\bf w}_t$ are distinct.
Tuple ${\bf w}_i$ will be called a {\em named flaw}. The flaw associated with the named flaw ${\bf f}$ will be denoted without the bold font as $f$,
i.e.\ ${\bf f}=(f,n_f)$. 
We say that a word $W={\bf w}_1\ldots {\bf w}_t$ over named flaws has {\em consistent counts} 
if it is constructed from the unnamed version ${w}_1\ldots {w}_t$ as described above.

%We will use named flaws as in the previous section, only now it will be convenient to assume that
%the counts of named flaws are computed from the left rather than from the right,
%and so e.g.\ ${\bf f}=(f,1)$ is the leftmost occurrence of flaw $f$ in a walk.
%A pair of named flaws $({\bf f},{\bf g})$ will be called a {\em swappable pair in a walk $\tau$}
%if $\tau=\ldots{\bf f}{\bf g}\ldots[\sigma_{t+1}]$ and $f\not\cong g$.
%
%We will denote the flaw fixed in the theorem as $\hat f$ (rather than $f$),
%and denote $\hat {\bf f}=(\hat f,1)$.
%
We will denote the element of $F$ fixed in Theorem~\ref{th:ParSwapping:f} as $\hat f$ (rather than $f$). 
In the case of Theorem~\ref{th:Swapping}(b) element $\hat f$ is not defined;
to indicate this fact, we will write $\hat f=\perp$.
%we assume that $\hat f=\varnothing$.

For a walk $\tau\myeq {\bf w}_1\ldots{\bf w}_t$ let us define a directed acyclic graph $\G_\tau=(\V_\tau,\E_\tau)$ as follows:
its nodes are $V_\tau=\{{\bf w}_1,\ldots,{\bf w}_t\}$,
and the set $\E_\tau$ contains all edges of the form $({\bf f},{\bf g})$
where $f\cong g$ and ${\bf f}$ occurs in $\tau$ before ${\bf g}$, i.e.\ $\tau\myeq\ldots{\bf f}\ldots{\bf g}\ldots$.
%Comparing this definition with the definition of the equivalence relation $\equiv$ in the previous section,
%it can be checked that $\tau\equiv\tau'$ if and only if $\G_\tau=\G_{\tau'}$. 

%\paragraph{Theorem~\ref{th:ParSwapping}} 
We say that ``walk $\tau$ contains $\hat f$'' if either (i) $\hat f= \perp$, or (ii)
$\hat f \ne \perp$ and  $\tau\myeq\ldots\hat f\ldots$.
For a walk $\tau$ containing $\hat f$ we make the following definitions:
\begin{itemize}
\item If $\hat f=\perp$ (i.e.\ in the case of Theorem~\ref{th:Swapping}(b)), let $\W_\tau=\V_\tau$.
For ${\bf f}\in\W_\tau$ let $d_\tau({\bf f})$ be the length of
the longest path from ${\bf f}$ to a sink of $\G_\tau$ plus one (so that for any sink node ${\bf f}\in\V_\tau$ we have $d_\tau({\bf f})=1$). 
\item If $\hat f\ne \perp$ (i.e.\ in the case of Theorem~\ref{th:ParSwapping:f}), let $\W_\tau$ be the set of flaws ${\bf g}\in \V_\tau$
from which node $\hat{\bf f}$ can be reached in $\G_\tau$,
where $\hat{\bf f}$ is the named flaw corresponding to the rightmost occurence of $\hat f$ in $\tau$.
For ${\bf f}\in\W_\tau$ let $d_\tau({\bf f})$ be the length of
the longest path from ${\bf f}$ to $\hat{\bf f}$ in $\G_\tau$ plus one (so that $d_\tau(\hat{\bf f})=1$). 
\end{itemize}
%We also define words $\Psi(\tau)$ and $\Psi_\pi(\tau)$ consisting of named flaws in $\W_\tau$ taken in the following order:
%\begin{itemize}
%\item The order 
%\item 
%\end{itemize}
%Let $\Psi(\tau)$ be the word consisting of the named flaws in $\W_\tau$ listed in
%the order in which they appear in~$\tau$. We also define another word $\Psi_\pi(\tau)$ as follows:
%Furthermore, we define words $\Psi(\tau)$ and $\Psi_\pi(\tau)$ as follows:
Furthermore, we define word $\Psi(\tau)$ as follows:
\begin{itemize}
%\item $\Psi(\tau)$ is the word consisting of the named flaws in $\W_\tau$ listed in
%the order in which they appear in~$\tau$. 
\item 
For an integer $r\ge 1$
 let $\I_r=\{{\bf f}\in\W_\tau\:|\:d_\tau({\bf f})=r\}$,
and let $W_r$ be the word consisting of the named flaws in $\I_r$ 
sorted in the decreasing order (with respect to $\preceq_\pi$). \footnote{It can be seen that we have $f\ne g$ for any distinct ${\bf f},{\bf g}\in\I_r$.
Therefore, ``sorting $\I_r$ in the decreasing order'' in the definition is a valid operation
(where we assume that  ${\bf f}\prec_\pi{\bf g}$ if $f\prec_\pi g$).
}
Then $\Psi(\tau)=W_s\ldots W_1$ where $s=\max \{ d_\tau({\bf f}) \:|\:{{\bf f}\in\W_\tau}\}$.
\end{itemize}
\begin{lemma}
The word $W=\Psi(\tau)$ is reversely $\pi$-stable.
Also, $R^{\tt rev}_W=\{\hat f\}$ if $\hat f\ne \perp$.
\end{lemma}
\begin{proof} 
Let $I_r=\{f\:|\:{\bf f}\in \I_r\}$ be the ``unnamed'' version of the set $\I_r$. % from Definition~\ref{def:Stabtau}.
%(or equivalently the set of flaws that appear in the word $W_r$).
Consider $r\in[s-1]$.
From the the definition of $\I_r$ and $\I_{r+1}$ we obtain that for every ${\bf f}\in \I_{r+1}$
there exists ${\bf g}\in \I_r$ with $({\bf f},{\bf g})\in E_\tau$ (implying that $f\in\Gamma^+(g)$).
This gives that $I_{r+1}\subseteq\Gamma^+(I_r)$. Also, if $\hat f\in F$ then we have $\I_1=\{\hat{\bf f}\}$ and so $I_1=\{\hat f\}$.
Observing that  $W=W_s\ldots W_1$ we
obtain the desired claims.
\end{proof}

It is straighforward to check that applying valid swaps to $\tau$ % containing $\hat f$ does
does not affect graph~$\G_\tau$, set~$\W_\tau$ and  word $\Psi(\tau)$. % (but may affect the order of named flaws in $\Psi(\tau)$).
Our goal will be to apply swaps to $\tau$ so that word $\Psi(\tau)$
becomes a prefix of $\tau$. %, and furthermore this prefix equals  $\Psi_\pi(\tau)$.
%\begin{definition}\label{def:Stabtau}
%For a walk $\tau$ containing $\hat f$ define word $\Psi_\pi(\tau)$ as follows.
%For an integer $r\ge 1$
% let $\I_r=\{{\bf f}\in\W_\tau\:|\:d_\tau({\bf f})=r\}$,
%and let $W_r$ be the word consisting of the named flaws in $\I_r$ 
%sorted in the decreasing order (with respect to $\preceq_\pi$).
%Then $\Psi_\pi(\tau)=W_s\ldots W_1$ where $s=\max \{ d_\tau({\bf f}) \:|\:{{\bf f}\in\W_\tau}\}$.
%\end{definition}
%
%Our next goal is to show how to transform walk $\tau\in\calX^{\hat f}$
%to a walk $\tau'\myeq{\tt Stab}_\pi(\tau)U$.
We will do this by applying swaps to {\em swappable pairs} in $\tau$.
\begin{definition}
Consider a walk $\tau$ containing $\hat f$.
A pair $({\bf f},{\bf g})$ of named flaws is called a {\em swappable pair in $\tau$}
if it can be swapped in $\tau$ (i.e.\  $\tau\myeq\ldots{\bf f}{\bf g}\ldots$ and $f\not\cong g$)
and either \\
(i) $({\bf f},{\bf g})\in (\V_\tau-\W_\tau)\times \W_\tau$, or \\
(ii) $({\bf f},{\bf g})\in \W_\tau \times \W_\tau$ and their order in $\Psi(\tau)$ is different:
%${\tt Stab}_\pi(\tau)\!=\!\ldots{\bf g}\ldots{\bf f}\ldots$.
$\Psi(\tau)=\shortldots{\bf g}\shortldots{\bf f}\shortldots$.

The position of the rightmost swappable pair in $\tau$ will be denoted as $k(\tau)$,
where the position of $({\bf f},{\bf g})$ in $\tau$ is the number of named
flaws that precede ${\bf g}$ in $\tau$. If $\tau$ does not contain a swappable pair then $k(\tau)=0$.
Thus, $k(\tau)\in[0,|\tau|-1]$.
\end{definition}
It can be seen that the procedure of repeatedly applying swaps to swappable pairs in $\tau$
must terminate. Indeed, swapping a pair of the form (i)
moves a named flaw in $\W_\tau$ to the left, which can happen only a finite number of times.
Similarly, swapping a pair of the form (ii) decreases the number of pairs $\{{\bf f},{\bf g}\}\in \binom{\W_\tau}{2}$
whose relative order in $\tau$ is not consistent with the relative order in $\Psi(\tau)$,
which again can happen only a finite number of times. Note, swaps of the other type
do not affect these arguments.

\begin{lemma}\label{lemma:tauABC}
Consider a walk $\tau\myeq AU$ containing $\hat f$ 
where $A,U$ are some words, and there are no swappable pairs inside $U$.
Then $U=BC$ where sequence $B$
is a subsequence of $\Psi(\tau)$, and $C$ does not contain
named flaws from $\W_\tau$.

In particular, if $|A|=0$ and $\tau\myeq U$ does not contain a swappable pair 
then  $\tau\myeq \Psi(\tau)C$.
\end{lemma}
\begin{proof}
Let ${\bf u}_1,\ldots,{\bf u}_m$ be the named flaws from $\W_{\tau}$
that occur in $U$ (listed in the order of their appearance in $U$). 
We claim that ${\bf u}_1\ldots{\bf u}_m$ is a prefix of $U$.
Indeed, suppose not, then $U=\ldots {\bf f} {\bf u}_i\ldots$ where ${\bf f}\notin \W_\tau$
and ${\bf u}_i\in \W_\tau$. This means that $({\bf f},{\bf u}_i)\notin E_\tau$, and so $f\not\cong u_i$.
But then $({\bf f},{\bf u}_i)$ is a swappable pair in $U$ - a contradiction.

We obtained a decomposition $\tau\myeq ABC$ where $B={\bf u}_1\ldots{\bf u}_m$.
It remains to show that $B$ is a subsequence of $\Psi(\tau)$. For that it suffices
to prove that for any $i\in[m-1]$ the relative order of ${\bf u}_i$ and ${\bf u}_{i+1}$
in $B$ is the same as in $\Psi(\tau)$, i.e.\ $\Psi(\tau)=\ldots{\bf u}_{i}\ldots{\bf u}_{i+1}\ldots$.
We must have $u_i\cong u_{i+1}$ (otherwise $({\bf u}_i,{\bf u}_{i+1})$ would be a swappable
pair in $\tau$, contradicting the assumption). Therefore, $({\bf u}_{i},{\bf u}_{i+1})\in\E_\tau$,
implying that $d_\tau({\bf u}_{i})>d_\tau({\bf u}_{i+1})$.
Inspecting the definition of $\Psi(\tau)$, we conclude that ${\bf u}_i$ should be to the left of ${\bf u}_{i+1}$ in $\Psi(\tau)$,
which is what we needed to show.
\end{proof}

To summarize, we showed that taking a walk $\tau\in\calX$ and repeatedly applying swaps to swappable pairs
terminates and produces a walk $\tau'$ such that the word $\Psi(\tau')=\Psi(\tau)$ is a prefix of $\tau'$.
One consequence is that $\Psi(\tau)$ must have consistent counts
(since applying valid swaps preserves consistency of counts).
We can also conclude that there exists a swapping mapping $\Phi$ and set $\calX_\pi=\Phi(\calX)$
with the following properties: in the case of Theorem~\ref{th:Swapping}(b)
all walks $\tau\in\calX_\pi$ are reversely $\pi$-stable,
and in the case of Theorem~\ref{th:ParSwapping:f} set $\calX_\pi$ satisfies condition (i).
It remains to show that $\Phi$ can be chosen so that it is injective on $\calX$,
and in the case of Theorem~\ref{th:ParSwapping:f} set $\calX_\pi$ satisfies condition (ii).
For that we need to be careful
with the order in which we apply swaps to swappable pairs. 
We will use the following algorithm. %For a walk $\tau={\bf w}_1\ldots {\bf w}_k{\bf w}_{k+1} \ldots {\bf w}_t[\sigma_{t+1}]$
First, let $\calX_0=\calX$,
and then for $p=0,1,2,\ldots$ do the following:
% while $\calX_p$ contains a walk $\tau$ with a swappable pair:
\begin{itemize}
\item Let $k=\max_{\tau\in\calX_p}k(\tau)$. If $k=0$ then terminate.
\item For each $\tau\in\calX_p$ do the following: if $k(\tau)=k$
then swap the pair $({\bf f},{\bf g})$ at position $k$ in $\tau$, otherwise leave $\tau$ unchanged. Let $\calX_{p+1}$ be the new set of walks.
\end{itemize}

Note that this algorithm defines a mapping from $\calX_p$ to $\calX_{p+1}$ in a natural way.
Let us fix a word $W={\bf w}_1\ldots {\bf w}_t$ over named flaws with consistent counts, and define $\calX_p[W]=\{\tau\in\calX_p\:|\:\Psi(\tau)=W\}$ for an index $p\ge 0$.
Since applying valid swaps to $\tau$ does not affect $\Psi(\tau)$, the mapping defined by the algorithm above
sends walks in $\calX_p[W]$ to walks in $\calX_{p+1}[W]$.
The remaining claims will follow from the following result (note that the set $\calX_0[W]\subseteq\calX$ is valid by the assumption of the theorems). 

\begin{lemma}
If set $\calX_p[W]$ is valid then so is $\calX_{p+1}[W]$,
and the mapping from $\calX_p[W]$ to $\calX_{p+1}[W]$ defined by the algorithm above is injective.
\end{lemma}
\begin{proof}
Suppose that the lemma is false. This means that there exist  distinct
walks $\tau,\bar\tau\in \calX_p[W]$ that were transformed
to walks $\eta,\bar\eta \in \calX_{p+1}[W]$, respectively, such
that $\eta\wedge\bar\eta=\ldots[\Omega]$ (violating either injectiveness
if $\eta=\bar\eta$, or validity if $\eta\ne\bar\eta$).
At least one of the walks $\tau,\bar\tau$ must have changed;
assume w.l.o.g.\ that the walk $\tau$ has changed, i.e.\ $\eta\ne\tau$.
We thus have 
$$
\tau\myeq \xi\;{\bf f}\:{\bf g}\:X,\qquad
\eta\myeq \xi\;{\bf g}\:{\bf f}\:X
$$
where $\xi$ is some walk (of length $k-1=k(\tau)-1$), $({\bf f},{\bf g})$ is a swappable pair in $\tau$, and $X$ is some word over named flaws.
Note that $\bar\eta$ either equals $\bar\tau$ or is obtained from $\bar\tau$ by swapping a pair at position $k$ in $\bar\tau$.

We have $\tau\wedge\bar\tau\ne\ldots[\Omega]$ since $\calX_p[W]$
is valid.
If $\tau\wedge\bar\tau$ is a proper prefix of $\xi$ then $\eta\wedge\bar\eta=\tau\wedge\bar\tau$,
 contradicting the assumption that  $\eta\wedge\bar\eta= \ldots[\Omega]$.
Thus, $\xi$ is a prefix of $\tau\wedge\bar\tau$. Condition $\tau\wedge\bar\tau\ne\ldots[\Omega]$ means that
$$
\bar\tau\myeq \xi\;{\bf f}\:Y
$$
for some word $Y$ over named flaws.
By construction, walk $\tau$ does not have swappable pairs inside ${\bf g}X$,
and walk $\bar\tau$ does not have swappable pairs inside $Y$. Lemma~\ref{lemma:tauABC} gives
that ${\bf g}X=BC$ and $Y=\bar B\bar C$ where  \\
(i) $B,\bar B$ are subsequences of $\Psi(\tau)=\Psi(\bar\tau)=W$ (with $B={\bf g}\ldots$), and \\
(ii) $C,\bar C$ do not contain named flaws from $\W$
where we denoted $\W=\W_\tau=\W_{\bar\tau}$ (it is the set of named flaws present in $W$).

%(where $\W_\tau=\W_{\bar\tau}$ holds since $\W_\tau$ is the set of named flaws present in $\Psi(\tau)=W$,
%and similarly for $\W_{\bar\tau}$). Rec

We obtained that $\tau\myeq \xi\;{\bf f}\:B C$  and $\bar\tau\myeq \xi\;{\bf f}\:\bar B \bar C$.
The set of named flaws present in $B$ must be the same as the set of named flaws present in $\bar B$
(namely, it is $\W$ minus the set of named flaws present in the sequence $\xi\;{\bf f}$).
Furthermore, $B$ and $\bar B$ are subsequences of $W$, so we must have $B=\bar B$
(recall that all named flaws in $W$ are unique).
Therefore, $\bar\tau\myeq \xi\;{\bf f}\:{\bf g}\ldots$.
This means that $({\bf f},{\bf g})$ is a swappable pair in $\bar\tau$, and  $\bar\eta$ is obtained from $\bar\tau$
by swapping ${\bf f}$ and ${\bf g}$.

To summarize, we showed that walks $\tau,\bar\tau,\eta,\bar\eta$ have the following forms:
\begin{eqnarray*}
\tau = \xi[\sigma_1]{ f}[\sigma_2]{ g}[\sigma_3]\zeta &    \qquad\quad& \eta = \xi[\sigma_1]{g}[\sigma'_2]{f}[\sigma_3]\zeta \\
\bar\tau = \xi[\sigma_1]{f}[\bar\sigma_2]{g}[\bar\sigma_3]\bar\zeta & & \bar\eta = \xi[\sigma_1]{g}[\bar\sigma'_2]{f}[\bar\sigma_3]\bar\zeta
\end{eqnarray*}
where $\sigma_1$ is the last state of walk $\xi$, $\sigma_3$ is the first state of walk $\zeta$, and $\bar\sigma_3$ is the first state of walk $\bar\zeta$.
Condition $\eta\wedge\bar\eta=\ldots[\Omega]$ implies that $\bar\sigma'_2=\sigma'_2$ and $\bar\sigma_3=\sigma_3$.
Since the {\tt SWAP} operation from Definition~\ref{def:main} is injective, we obtain that $\bar\sigma_2=\sigma_2$.
It can now be seen that we cannot have simultaneously $\tau\wedge\bar\tau\ne \ldots[\Omega]$ and $\eta\wedge\bar\eta=\ldots[\Omega]$.
We have obtained a contradiction.
\end{proof}

%%%%%%%%%%%%%%%%%%%%%%%%%%%%%%%%%%%%%%%%%%%%%%%%%%%%%%%%%%%%%%%%%%%%%%%%%%%%%%%%%%%%%%%%%%%%%%%%%%%%%%%%
%%%%%%%%%%%%%%%%%%%%%%%%%%%%%%%%%%%%%%%%%%%%%%%%%%%%%%%%%%%%%%%%%%%%%%%%%%%%%%%%%%%%%%%%%%%%%%%%%%%%%%%%
%%%%%%%%%%%%%%%%%%%%%%%%%%%%%%%%%%%%%%%%%%%%%%%%%%%%%%%%%%%%%%%%%%%%%%%%%%%%%%%%%%%%%%%%%%%%%%%%%%%%%%%%
%%%%%%%%%%%%%%%%%%%%%%%%%%%%%%%%%%%%%%%%%%%%%%%%%%%%%%%%%%%%%%%%%%%%%%%%%%%%%%%%%%%%%%%%%%%%%%%%%%%%%%%%

\section{Examples of commutative resampling oracles for non-Cartesian spaces}\label{sec:construction}
%\section{Constructing atomic commutative resampling oracles}\label{sec:construction}
In this section we show that resampling oracles for permutations used in~\cite{Harris:permutations,HarveyVondrak15}
and for perfect matchings in complete graphs used in~\cite{HarveyVondrak15} are strongly commutative.

These oracles have two important properties:
(i) they are atomic;
(ii)  $\rho(\sigma'|f,\sigma)$ are uniform distributions over $A(f,\sigma)$ for any $f\in F$ and $\sigma\in f$.
Distributions $\rho$ satisfying these two properties were studied
by Achlioptas and Iliopoulos~\cite{Achlioptas}, albeit with a different terminology:
they used a directed multigraph $D$ instead of $\rho$. 
This multigraph is defined as follows:
%Let us define a directed multigraph $D$ as follows:
its set of nodes is $\Omega$, and its set of edges is the set of all valid walks $\sigma\RA{f}\sigma'$.
(Each edge of $D$ is labeled by a flaw in $F$). It is {\em atomic}
if any state $\sigma'\in \Omega$ has at most one incoming edge in $D$ labeled by a given flaw  $f\in F$.
Note that one can recover sets $A(f,\sigma)$ from $D$,
since $A(f,\sigma)=\{\sigma'\:|\:\mbox{edge }\sigma\RA{f}\sigma'$ $\mbox{belongs to $D$}\}$.
Therefore, distributions $\rho$ are uniquely defined by $D$, assuming that property (ii) holds.
It can be seen that there is a one-to-one correspondence between atomic multigraphs and distributions $\rho$
satisfying (i) and (ii).

%Instead of doing this directly, we find it easier to do the following:
To prove commutativity of the oracles mentioned earlier, we find it easier to take an indirect approach: first, we will
describe a generic route for constructing atomic multigraphs,
then apply it to permutations and matchings and prove weak commutativity.
We will then see that the resulting resampling procedure coincides with that in~\cite{Harris:permutations,HarveyVondrak15}.

The constructed multigraph will satisfy the following property: for each flaw $f$
there exists constant $A_f$ such that 
$|A(f,\sigma)|=A_f$ all $\sigma\in f$.
Thus, $\rho(\sigma'|f,\sigma)=\frac{1}{A_f}$ for $\sigma'\in A(f,\sigma)$, and so weak commutativity will imply strong commutativity.

We start with some general observations.
For an atomic multigraph $D$ let us define a mapping $\psi:F\times \Omega\rightarrow\Omega\cup\{{\perp}\}$ 
that specifies a ``backward step''  for $f\in F$  and $\sigma'\in \Omega$ as follows:
if there exists a (unique) state $\sigma\in \Omega$ such that $\sigma\stackrel{f}\rightarrow\sigma'$ 
then $\psi(f,\sigma')=\sigma$, otherwise $\psi(f,\sigma')={\perp}$. Note that $\psi$   satisfies the following properties:
\begin{itemize}
\item[(I)] If $\psi(f,\sigma')\ne{\perp}$ then  $\psi(f,\sigma')\in f$.
\item[(II)] For any $f\in F$ and $\sigma\in f$ there exists at least one $\sigma'\in\Omega$ with $\psi(f,\sigma')=\sigma$.
\end{itemize}

It is not difficult to see that $D$ can be uniquely reconstructed from $\Omega$, $F$ and $\psi$,
since for each $f\in F$ and $\sigma\in f$ we have $A(f,\sigma)=\{\sigma'\in\Omega\:|\:\psi(f,\sigma')=\sigma\}$.
Furthermore, any triplet $(\Omega,F,\psi)$ specifies a valid atomic multigraph $D$, as long  as $\psi$ satisfies properties (I) and (II)
(and $F$ is some set of non-empty subsets of $\Omega$). Property (II), in particular, is equivalent to the condition that  $A(f,\sigma)$ is non-empty for each $f\in F$ and $\sigma\in f$.

Thus, the problem of constructing the set of actions $A(f,\sigma)$ for a given $f,\sigma$ can be shifted to the problem of constructing a mapping $\psi$.
Of course, after constructing $\psi$ one still needs to show that sampling from $A(f,\sigma)$ can be done efficiently.
%Next, we apply this approach to perfect matching problems.

We remark that Definition~\ref{def:causality} of the potential causality graph can
also be reformulated in terms of the mapping $\psi$, as stated below (this claim follows directly from definitions).
\begin{proposition}
Undirected graph $(F,\sim)$ is a potential
causality graph for $D$ if
for any $f,g\in F$ with $f\nsim g$ and any $\sigma'\in g$ with $\psi(f,\sigma')\ne{\perp}$
we have $\psi(f,\sigma')\in g$ and $f\ne g$.
\label{prop:psiCausality}
\end{proposition}
%(If we had $\psi(f,\sigma)\notin g$ or $f\ne g$,
%then addressing $f$ could cause to appear, in which case we would have $f\sim g$).

\subsection{Matchings}
We now apply the route outlined above to some matching problems.
Let $G=(V,E)$ be an undirected graph with $|V|=2n$ nodes that satisfies the following condition:
\begin{itemize}
\item[($\ast$)]
{\em
If $(u',u,v,v')$ is a path in $G$ with distinct nodes then $\{u',v'\}\in E$.
}
\end{itemize}
We will consider the case when $\Omega$ is the set of perfect matchings in $G$
(so that each object $\sigma\in\Omega$ is a subset of $E$).
We allow any flaw of the form $f_M=\{\sigma\in\Omega\:|\:M\subseteq \sigma\}$
where $M$ is a fixed subset of $E$. 
It can be  assumed w.l.o.g.\ that $M$ is a matching (otherwise $f_M$ would be empty).
Thus, $F$ can be any subset of $\{f_M\:|\:M\in\calM\}$ where $\calM$ denotes
the set of matchings in $G$, with $\Omega\subset\calM$.
%(As we will see later, $f_M$ is non-empty for any $M\in\calM$).
%
%For two flaws $f_M,f_{M'}\in F$ we define $f_M\nsim f_{M'}$ iff $M\cap M'=\varnothing$ and $M\cup M'\in\calM$.
%
Two special cases of this framework have been considered~\cite{Harris:permutations,Achlioptas,HarveyVondrak15}:
\begin{itemize}
\item[{[P1]}] $G$ is the complete graph on $2n$ vertices, so that $\Omega$ is the set of all perfect matchings of $V$.
%This can model e.g.\ the {\tt Rainbow Matching} problem considered in.
\item[{[P2]}] Set $V$ can be partitioned into disjoint subsets $A_1,B_1,\ldots,A_r,B_r$ such
that $|A_i|=|B_i|$ for $i\in[r]$ and $E=\{\{u,v\}\:|\:u\in A_i,v\in B_i,i\in[r]\}$.
Thus, $G$ is a union of $r$ complete bipartite graphs, and set $\Omega$  corresponds to $r$ permutations.
\end{itemize}
In fact, these are essentially the only possibilities allowed by condition ($\ast$):
it can be shown that if $G$  contains at least one perfect matching then
each component of $G$ is either a complete graph $K_{m}$ or a complete bipartite graph $K_{m,m}$.
We will not need this claim, and so we leave it without a proof. Instead, we just assume
that one of the two cases [P1,P2] holds (but unlike previous work, we will treat them in a unified way,
relying mostly on condition ($\ast$)).

We use the following potential causality graph:
 $f_M\sim f_{M'}$ for $f_M,f_{M'}\in F$ if $M\cup M'$ is not a matching
{\em or} $M=M'$.
This graph is the same or slightly smaller than graphs used previously
for cases [P1] and [P2]. (For [P2] the works~\cite{Harris:permutations,HarveyVondrak15}  used a larger relation $\sim'$ instead where 
 $f_M\sim' f_{M'}$ if $M\cup M'$ is  not a matching {\em or} $M\cap M'\ne \varnothing$).

%Next, we will discuss how to design atomic multigraphs that
%are commutative with respect to~$\sim$. In the two special
%cases above we will arrive at the same resampling oracles used in~\cite{Harris:permutations,HarveyVondrak15}.

%Let $\calM=\{M\subseteq E\:|\:M\ne\varnothing,\mbox{distinct edges in $M$ are node-disjoint}\}$.
We will construct a mapping $\hat\psi:\calM\times \Omega\rightarrow\Omega$ that satisfies 
$M\subseteq\hat\psi(M,\sigma)$ for any $M\in\calM$ and $\sigma\in\Omega$.
It will correspond to the mapping $\psi$ in a natural way, i.e.\ $\psi(f_M,\sigma)=\hat\psi(M,\sigma)$ for $f_M\in F$.
Clearly, such $\psi$ will be defined everywhere on $F\times \Omega$, and will satisfy property (I).
%We will also get as a consequence that for any $M\in\calM$ the set $f_M$ is non-empty
%(as we claimed earlier), since $\hat\psi(M,\sigma)\in f_M$ for any $\sigma\in\Omega$.

\paragraph{Defining $\hat\psi$}
Consider $M\in\calM$ and $\sigma\in\Omega$. If $M$ contains a single edge $e=\{u,v\}$,
then we find unique $u',v'$ with $\{u,u'\},\{v,v'\}\in \sigma$, and set 
\begin{equation}
\hat\psi(\{e\},\sigma)=(\sigma-\{\{u,u'\},\{v,v'\}\})\;\cup\;\{\{u,v\},\{u',v'\}\}
\end{equation}
Note that if $\{u,v\}\in\sigma'$ then $\hat\psi(M,\sigma)=\sigma$. Otherwise nodes $u,v,u',v'$
are distinct, and we have $\hat\psi(M,\sigma)\in\Omega$
by the assumption ($\ast$) on the graph $G$.

Now suppose that $M=\{e_1,\ldots,e_k\}$ contains more than one edge.
Then we define 
\begin{equation}
\sigma_0=\sigma\qquad
\sigma_1=\hat\psi(\{e_1\},\sigma_0)\quad
\sigma_2=\hat\psi(\{e_2\},\sigma_1)\quad
\ldots
\quad
\sigma_{k}=\hat\psi(\{e_{k}\},\sigma_{k-1})\quad
\label{eq:psiDef}
\end{equation}
and set $\hat\psi(M,\sigma)=\sigma_{k}$. To show that this is well-defined, we need to prove
that the result does not depend on the chosen ordering of $M$. It suffices to prove this claim for $|M|=2$,
then we can use an induction argument (since any ordering of $M$ can be transformed to any other ordering via a sequence
of operations that swap adjacent elements).
Proving it for $|M|=2$ can be done by inspecting all possible cases, which are visualized in Fig.~\ref{fig:matchings}; verification
of the claim in each case is left to the reader.

\begin{figure}[t]
\vskip 0.2in
\normalsize
\begin{center}
\begin{tabular}{c}\hspace{-9pt}
\includegraphics[scale=0.5]{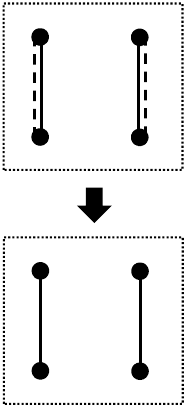}\,
\includegraphics[scale=0.5]{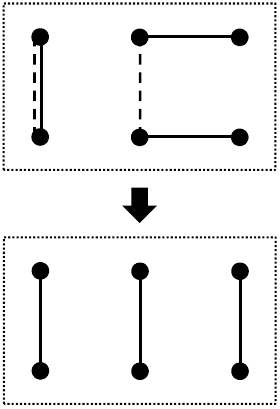}\,
\includegraphics[scale=0.5]{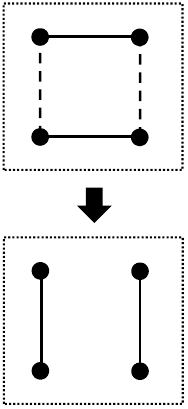}\,
\includegraphics[scale=0.5]{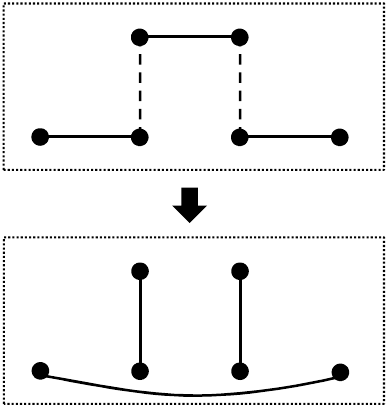}\,
\includegraphics[scale=0.5]{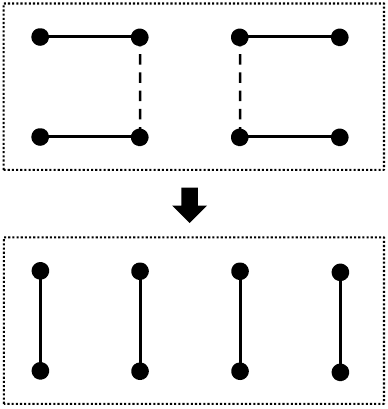}
\end{tabular}
\caption{Possible cases for $M=\{e_1,e_2\}\in\calM$ and $\sigma\in\Omega$.
Top row: solid lines indicate edges in $\sigma$, dashed lines indicate edges in $M$.
Bottom row: result of $\psi(\{e_1\},\psi(\{e_2\},\sigma))=\psi(\{e_2\},\psi(\{e_1\},\sigma))$.
}
\label{fig:matchings}
\end{center}
\vskip -0.2in
\end{figure}

\begin{proposition}
(a) In~\eqref{eq:psiDef} we have  $\{e_1,\ldots,e_i\}\subseteq\sigma_i$ for any $i\in[k]$.
Consequently, $M\subseteq\hat\psi(M,\sigma)$ for any $M\in\calM$ and $\sigma\in\Omega$ (and so property (I) holds). \\
(b) For any $\sigma\in\Omega$ and $M,M'\in \calM$ with $M\cup M'\in\calM$  we have 
$\hat\psi(M,\hat\psi(M',\sigma))=\hat\psi(M',\hat\psi(M,\sigma))$.
Consequently, the multigraph $D$ defined by $(\Omega,F,\psi)$ is weakly commutative with respect to $\sim$. \\
(c) Relation $\sim$ is a potential causality graph for $D$.
\label{prop:psiDef}
\end{proposition}
\begin{proof}
\noindent{\bf (a)~~} 
%It can be checked using induction that in~\eqref{eq:psiDef} we have
% $\{e_1,\ldots,e_i\}\subseteq\sigma_i$ for any $i\in[k]$. 
We can use induction on $i$. The base case $i=0$ is
vacuous. The induction step for $i$ follows from the definition of the mapping $\hat\psi(\{e_i\},\sigma_{i-1})$
and the fact that $\sigma_{i-1}$ cannot contain edges that connect 
an endpoint of $e_i$ to an endpoint of $e_j$ for $j\in[i-1]$
(since $\{e_1,\ldots,e_{i-1}\}\subseteq\sigma_{i-1}$ by the induction hypothesis, and $\{e_1,\ldots,e_i\}$
is a matching).

\noindent{\bf (b)~~} We claim $\hat\psi(M',\hat\psi(M,\sigma))=\hat\psi(M\cup M',\sigma)$.
Indeed, let $M=\{e_1,\ldots,e_\ell\}$ and $M'=\{e_{\ell+1},\ldots,e_k\}$.
Define $\sigma_0,\sigma_1,\ldots,\sigma_k$ as in~\eqref{eq:psiDef},
then  $\hat\psi(M,\sigma)=\sigma_\ell$ and $\hat\psi(M',\sigma_\ell)=\sigma_k$.
Thus, it remains to show that $\hat\psi(M\cup M',\sigma)=\sigma_k$.
For that we need to observe that if some edge $e_i$ appears in
the sequence $(e_1,\ldots,e_k)$ the second time then we have
$e_i\in\sigma_{i-1}$ and consequently $\sigma_i=\hat\psi(\{e_i\},\sigma_{i-1})=\sigma_{i-1}$.
Thus, such $e_i$ can be removed from the sequence without affecting the result.
After removing duplicates we conclude that $\hat\psi(M\cup M',\sigma)=\sigma_k$ by the definition of $\hat\psi$.

In a similar way we can show that $\hat\psi(M,\hat\psi(M',\sigma))=\hat\psi(M\cup M',\sigma)$. This proves the claim.

\noindent{\bf (c)~~} Let us show that conditions of Proposition~\ref{prop:psiCausality} hold.
Consider flaws $f=f_M$, $g=f_{M'}$ in $F$ with $f_M\nsim f_{M'}$ and object $\sigma\in f_{M'}$.
Condition $f_M\ne f_{M'}$ holds since  $f_M\nsim f_{M'}$, so we need 
 to show that $\psi(f_M,\sigma)\in f_{M'}$, or equivalently $M'\subseteq \hat\psi(M,\sigma)$.

Assume that $M=\{e_1,\ldots,e_k\}$, and define sequence $\sigma_0,\sigma_1,\ldots,\sigma_{k-1},\sigma_k\!=\!\hat\psi(M,\sigma)$ as in~\eqref{eq:psiDef}.
Indeed, for the base case the claim $M'\subseteq\sigma$ holds since $\sigma\in f_{M'}$,
and for the induction step we need to use the definition of $\hat\psi$ 
and the fact that  $M\cup M'\in\calM$ (which holds since $f_M\nsim f_{M'}$).
We leave verification of the induction step to the reader.
\end{proof}

\paragraph{Sampling from $A(f_M,\sigma_k)$} The general idea is to ``reverse''
the process in eq.~\eqref{eq:psiDef}: given flaw $f_M$ with $M=\{e_1,\ldots,e_k\}\in\calM$
and object $\sigma_k\in f_M$,
we first generate possible values for $\sigma_{k-1}$, then for $\sigma_{k-2}$, and so on.
%One of the constraints that should be maintained is $

For a subset $S\subseteq E$ let $\overrightarrow S=\{(u,v),(v,u)\:|\:\{u,v\}\in A\}$ be a ``directed copy'' of $S$.
For an object $\sigma\in\Omega$ and edges $(u,v),(u',v')\in \overrightarrow \sigma$ define
\begin{eqnarray}
{\tt Swap}_\sigma((u,v),(u',v'))=(\sigma - \{\{u,v\},\{u',v'\}\})\;\cup\;\{\{u,u'\},\{v,v'\}\}
\end{eqnarray}
Finally, for an object $\sigma\in\Omega$ and an edge $(u,v)\in\overrightarrow \sigma$ let us define
\begin{eqnarray}
\calN_\sigma(u,v)=\{(u',v')\in  \overrightarrow\sigma\:|\:{\tt Swap}_\sigma((u,v),(u',v'))\in\Omega\}
\end{eqnarray}
It can be checked that $(v,u)\in\calN_\sigma(u,v)$ and $(u,v)\notin\calN_\sigma(u,v)$. Furthermore, in the special cases
above we have the following:

\begin{itemize}
\item[{[P1]}] $\calN_\sigma(u,v)=\overrightarrow\sigma-\{(u,v)\}$.
\item[{[P2]}] If $(u,v)\in A_i\times B_i$ then $\calN_\sigma(u,v)=(B_i\times A_i)\cap \overrightarrow\sigma$.
\end{itemize}

We can now formulate the sampling algorithm (see Algorithm~\ref{alg:sampling}).

\begin{algorithm}[!h]
\caption{Sampling from $A(f_M,\sigma_k)$ for $M=\{e_1,\ldots,e_k\}\subseteq\sigma_k\in\Omega$}\label{alg:sampling}
\begin{algorithmic}[1]
\STATE {\bf for} $i=k,k-1,\ldots,1$ {\bf do}
\STATE ~~~~choose orientation $(u,v)$ of edge $e_i=\{u,v\}$
\STATE ~~~~select $(u',v')\in\calN_{\sigma_i}(u,v)-\overrightarrow{\{e_1,\ldots,e_{i-1}\}}$ uniformly at random
\STATE ~~~~set $\sigma_{i-1}={\tt Swap}_{\sigma_i}((u,v),(u',v'))$
\STATE {\bf end for}
\STATE return $\sigma_0$
\end{algorithmic}
\end{algorithm}

Let us verify the correctness of this algorithm.
Using the definitions of $\hat\psi$ and ${\tt Swap}_\sigma$, the following fact can be easily checked.
\begin{lemma}
(a) Suppose that $\{u,v\}\in\sigma_i\in\Omega$, $(u',v')\in\calN_{\sigma_i}(u,v)$ and $\sigma_{i-1}={\tt Swap}_{\sigma_i}((u,v),(u',v'))$.
Then $\hat\psi(\{\{u,v\}\},\sigma_{i-1})=\sigma_i$. \\
(b) Conversely, suppose that $\hat\psi(\{\{u,v\}\},\sigma_{i-1})=\sigma_i$ for $e_i=\{u,v\}\in E$ and $\sigma_{i-1}\in\Omega$.
Then there exists unique $(u',v')\in\calN_{\sigma_i}(u,v)$ such that $\sigma_{i-1}={\tt Swap}_{\sigma_i}((u,v),(u',v'))$.
Furthermore, it satisfies $\{u',v'\}\notin{\sigma_{i-1}-\{e_i\}}$.
\label{lemma:fb}
\end{lemma}

Using this lemma, we can now show establish correctness of the sampling procedure.
We say that two executions of Algorithm~\ref{alg:sampling} are distinct if they
made different choices in line 3 for some $i\in[k]$.

\begin{proposition}
Algorithm~\ref{alg:sampling} is well-defined, i.e.\ in line 3 we have $e_i\in\sigma_i$.
It can generate object $\sigma_0\in\Omega$ if and only if $\hat\psi(M,\sigma_0)=\sigma_k$.
Finally, distinct executions produce distinct outputs.
\label{prop:sampling:execution}
\end{proposition}
\begin{proof} The proof will have two parts corresponding to two directions.

\noindent{\bf (a)~~} Let $\sigma_k,\sigma_{k-1},\ldots$ be the sequence of objects
produced by the algorithm.
We will show using induction on $i=k,\ldots,1,0$ that
 $\{e_1,\ldots,e_{i}\}\subseteq\sigma_{i}$ (and therefore line 3 for index $i$ is well-defined)
and $\hat\psi(\{e_{i+1},\ldots,e_k\},\sigma_{i})=\sigma_k$. The base case $i=k$ is trivial.
Suppose the claim holds for $i\in[k]$, let us show it for $i-1$.
We have  $\{e_1,\ldots,e_{i}\}\subseteq\sigma_{i}$ by the induction hypothesis;
inspecting the rule for choosing $(u',v')$, we conclude that $\{e_1,\ldots,e_{i-1}\}\subseteq\sigma_{i-1}$.
For the second claim we can write
$$
\hat\psi(\{e_i,\ldots,e_k\},\sigma_{i-1})
=\hat\psi(\{e_{i+1},\ldots,e_k\},\hat\psi( \{e_i\}     ,\sigma_{i-1}))
=\hat\psi(\{e_{i+1},\ldots,e_k\},\sigma_i)=\sigma_k
$$
where the first equality is by the definition of $\hat\psi$, the second is by Lemma~\ref{lemma:fb}(a) and third is by the induction hypothesis.
This concludes the argument.

\noindent{\bf (b)~~} 
Suppose that $\hat\psi(M,\sigma_0)=\sigma_k$ for $\sigma_0\in\Omega$ and $M=\{e_1,\ldots,e_k\}$.
Define objects $\sigma_1,\ldots,\sigma_k$ as in~\eqref{eq:psiDef}.
We claim that Algorithm~\ref{alg:sampling} can replicate this sequence (in the reverse order).
Indeed, by Lemma~\ref{lemma:fb}(b) it suffices to show that
for any $(u',v')\in \calN_{\sigma_i}(u,v)$ with $\{u',v'\}\notin{\sigma_{i-1}-\{e_i\}}$
we also have $(u',v')\in\calN_{\sigma_i}(u,v)-\overrightarrow{\{e_1,\ldots,e_{i-1}\}}$.
Suppose not, then $(u',v')\in \overrightarrow{\{e_1,\ldots,e_{i-1}\}}$.
By Proposition~\ref{prop:psiDef}(a) we have ${\{e_1,\ldots,e_{i-1}\}}\subseteq {\sigma_{i-1}}$,
and so $\{u',v'\}\in\sigma_{i-1}$. Thus,  $\{u',v'\}=e_i$.
But $e_i$ does not appear in $\{e_1,\ldots,e_{i-1}\}$, and so we cannot have $(u',v')\in \overrightarrow{\{e_1,\ldots,e_{i-1}\}}$ - a contradiction.

Let us now prove that the input $M$, $\sigma_k$ and the output $\sigma_0$
uniquely determine choices made during the execution (this will give the last claim of the lemma).
Let $\tilde\sigma_{k},\ldots,\tilde\sigma_1,\tilde\sigma_0$ be the objects produced during the execution,
with $\tilde\sigma_k=\sigma_k$ and $\tilde\sigma_0=\sigma_0$.
Set $i=1$.
By Lemma~\ref{lemma:fb}(a) we have $\hat\psi(\{e_i\},\sigma_{i-1})=\tilde\sigma_i$,
implying that $\tilde\sigma_i=\sigma_i$ is determined uniquely.
By Lemma~\ref{lemma:fb}(b) the choice of $(u',v')$ in line 3 for index $i$
is also determined uniquely from $\sigma_{i-1},\sigma_i$ and $e_i$.
Repeating this argument for $i=2,\ldots,k$ (i.e.\ using induction) yields the claim.
\end{proof}
We have proved that the output of Algorithm~\ref{alg:sampling} is a distribution whose support is $A(f_M,\sigma_k)$.
To show that this distribution is uniform, we need to observe additionally
that the number of choices in line 3 for index $i$ depends on $i$ but not on the past execution history (which
can be easily checked for cases [P1] and [P2]).
The cardinality of $A(f_M,\sigma_k)$ is the product of these numbers over $i\in[k]$,
and thus depends only on the flaw $f_M$ (more precisely, on $|M|$).

To summarize, we have constructed an atomic weakly commutative multigraph $D$,
proved that Algorithm~\ref{alg:sampling} samples uniformly from $A(f_M,\sigma_k)$,
and the size of latter set depends only on $f_M$ (the latter implies strong commutativity).
It can now be verified that the sampling procedure coincides with the procedure
in~\cite{HarveyVondrak15} for perfect matchings in a complete graph (in the case [P1]),
and 
with the procedure
in~\cite{Harris:permutations,HarveyVondrak15} for permutations (in the case [P2]).

\subsection{Application: rainbow matchings in complete graphs}\label{sec:rainbow}
We refer to~\cite{Harris:permutations,Achlioptas,HarveyVondrak15} for applications of resampling oracles for permutations and perfect matchings.
Here we revisit just one application, namely a {\em rainbow matching} problem.
Our primary goal is to demonstrate how the choice of the distribution $\omegainit$
affects the bound on the expected runtime, and also compare it with the parallel version.
% that expression~\eqref{eq:Tseq:b} can be smaller than~\eqref{eq:Tseq:a}
%for any $\sigmainit$,
%and therefore sampling $\sigmainit\in\Omega$ uniformly at random can have
%a better bound on the expected runtime compared to other initializations.

Let $G=(V,E)$ be a complete graph on $2n$ vertices such that each edge is assigned a color,
and each color appears in at most $q$ edges.
A perfect matching in $G$ is called {\em rainbow} if its edges have distinct colors.
Achlioptas and Iliopoulos~\cite{Achlioptas} showed  that a rainbow matching
exists 
if $q\le \gamma n$ for some constant $\gamma<\frac{1}{2e}\simeq 0.184$.
Instead of~\eqref{eq:Condition:b}, they used a stronger condition~\eqref{eq:Condition:b'}.
Harvey and Vondr\'ak~\cite{HarveyVondrak15} improved the constant to $\gamma=0.21$
by exploiting a condition with the cluster expansion correction analogous to~\eqref{eq:Condition}. Below we redo their calculations.

Let $F$ be the set of flaws $f_{M}$ such $M$ contains two vertex-disjoint edges of the same color,
and assume that we use the multigraph and relation $\sim$ constructed in the previous section.
\begin{proposition}
If $\gamma=0.21$ then condition~\eqref{eq:Condition} can be satisfied by setting $\mu_f=\mu=\frac{3}{4n^2}$
for $f\in F$
(if $n$ is sufficiently large).
\end{proposition}
\begin{proof}
Consider flaw $f_M$ where $M=\{\{v_1,v_2\},\{v_3,v_4\}\}\in F$.
For node $v\in V$ let $\Gamma(v)\subseteq F$ be the set of flaws $f_{M'}$ such that
at least one of the edges in $M'$ is incident to $v$. We have $|\Gamma(v)|\le (2n-1)(q-1)$
(there are at most $2n-1$ choices for the node $v'\in V$ matched to $v$,
and then at most $q-1$ choices for the second edge of the same color as $\{v,v'\}$).
It can be seen that
$\Gamma(f_M)\subseteq\Gamma(v_1)\cup\Gamma(v_2)\cup\Gamma(v_3)\cup\Gamma(v_4)$.
Furthermore, any independent subset
 $S\in\IND(f_M)$ can be formed by selecting at most one flaw from each of $\Gamma(v_i)$ for $i\in[4]$.
Therefore,
$$
    \sum_{S\in\IND(f_M)}\mu(S) 
\le \prod_{i=1}^4 (\mu^0 + \underbrace{\mu^1+\ldots+\mu^1}_{|\Gamma(v_i)|\mbox{\footnotesize ~times}})
\le (1 + (2n-1)(q-1)\mu)^4
$$
By inspecting Algorithm~\ref{alg:sampling}
we can conclude that $A_f=(2n-3)(2n-1)$ for each $f\in F$. Thus, we get the following condition:
there must exist $\mu>0$ such that expression
$$
    \theta=\frac{1}{(2n-3)(2n-1) \mu}\cdot (1+(2n-1)(q-1)\mu)^4
$$
is a constant smaller than 1.
Denote $\beta=(2n-3)(2n-1) \mu$, then
$$
\theta \le 
\frac{1}\beta \cdot \left(1 + 2n\cdot (\gamma n)\cdot\frac{\beta}{4n^2}\cdot(1+o(1))\right)^4
=\frac{(1+\frac{1}{2}\gamma\beta+o(1))^4}{\beta}
$$
The last expression will be smaller than 1 (for a sufficiently large $n$) if $\beta=3$ and $\gamma=0.21$,
where we used the constants from~\cite{HarveyVondrak15}.
\end{proof}

Let us now estimate the expression in~\eqref{eq:Tseq}. % and~\eqref{eq:Tseq:b}.
We have $|\Omega|=(2n-1)!!$ and $\log|\Omega|=\Theta(n\log n)$.
When ${\tt supp}(\omegainit)=\{\sigmainit\}$ for some $\sigmainit\in\Omega$, we get $\gammainit=|\Omega|$ and thus
 $T=\Omega(n\log n)$. If, on the other hand, $\omegainit=\omega$ then
we can write
$$
  \sum\limits_{R\in\bigcup_{\sigma\in\Omega}\Ind(F_{\sigma})}\mu(R) 
\le   \sum\limits_{R\subseteq F}\mu(R) = \prod_{f\in F} (1+\mu_f)=(1+\mu)^{|F|}
$$
Observing that $|F|\le (2n)^2q=O(n^3)$ and $\mu=O(1/n^2)$, we obtain that 
%in the case~\eqref{eq:Tseq:b} we have 
$
T=O(|F|\log(1+\mu))=O(n^3\log(1+O(\frac{1}{n^2})))=O(n^3\cdot \frac{1}{n^2})=O(n)
$.
Thus, choosing $\omegainit=\omega$ leads to a better bound than initializing the algorithm with some fixed state $\sigmainit$.
This may not be surprising, given that in the latter situation we need to consider the worst case.
Note that a linear bound on the expected
number of resampling steps has also  been shown in~\cite{HarveyVondrak15}.

We can also compute a bound on the expected number of rounds of the parallel version (Algorithm~\ref{alg:parallel}), assuming that $\omegainit=\omega$.
From Theorem~\ref{th:parallel} we get
$$
\mbox{$
T=O(\log \sum_{f\in F}\mu_f)=O(\log (|F|\cdot \mu))=O(\log (n^3\cdot \frac{1}{n^2}))=O(\log n)
$}
$$
It should be noted, however, that at the moment it is not
known whether a round of Algorithm~\ref{alg:parallel} 
can be implemented efficiently (i.e.\ in an expected polylogarithmic time)
for matchings in a complete graph.
Such implementation has only been shown for permutations~\cite{Harris:permutations}.
We conjecture that the technique in~\cite{Harris:permutations}
can be extended to matchings in a complete graph, but leave this question outside the scope of this work.

%%%%%%%%%%%%%%%%%%%%%%%%%%%%%%%%%%%%%%%%%%%%%%%%%%%%%%%%%%%%%%%%%%%%%%%%%%%%%%%%%%%%%%%%%%%%%%%%%%%%%%%%%%%%%%%%%%%%%
%%%%%%%%%%%%%%%%%%%%%%%%%%%%%%%%%%%%%%%%%%%%%%%%%%%%%%%%%%%%%%%%%%%%%%%%%%%%%%%%%%%%%%%%%%%%%%%%%%%%%%%%%%%%%%%%%%%%%
%%%%%%%%%%%%%%%%%%%%%%%%%%%%%%%%%%%%%%%%%%%%%%%%%%%%%%%%%%%%%%%%%%%%%%%%%%%%%%%%%%%%%%%%%%%%%%%%%%%%%%%%%%%%%%%%%%%%%
%%%%%%%%%%%%%%%%%%%%%%%%%%%%%%%%%%%%%%%%%%%%%%%%%%%%%%%%%%%%%%%%%%%%%%%%%%%%%%%%%%%%%%%%%%%%%%%%%%%%%%%%%%%%%%%%%%%%%
%%%%%%%%%%%%%%%%%%%%%%%%%%%%%%%%%%%%%%%%%%%%%%%%%%%%%%%%%%%%%%%%%%%%%%%%%%%%%%%%%%%%%%%%%%%%%%%%%%%%%%%%%%%%%%%%%%%%%

\appendix

%\appendix

\section{Counting stable sequences: Proof of Theorems~\ref{th:StabCounting} and~\ref{th:StabCounting:rev}}\label{sec:proof:StabCounting}
We focus on the proof of Theorem~\ref{th:StabCounting}; the proof of~\ref{th:StabCounting:rev}
will be analogous. 
As mentioned in Section~\ref{sec:stable:definitions},
there is just one essential difference between formulations of Theorem~\ref{th:StabCounting} and~\ref{th:StabCounting:rev}:
the condition ``there exists walk $\tau$ with $\tau\myeq W$''
is present in the definitions of both ${\tt Stab}_\pi$ and ${\tt Stab}^{\tt rev}_\pi$,
and has not been ``reversed''. We use this condition just once, namely in the proof of Proposition~\ref{prop:StronglyStable}
below, where the two cases are stated separately.

%Recall that in Section~\ref{sec:stable:definitions}
%we defined ${\tt Stab}_\pi$ as the set of $\pi$-stable words $W$ for which there exists a walk $\tau$ with $\tau\myeq W$.
%Note that set ${\tt Stab}^{\tt rev}_\pi$ does not necessarily equal the reverse of words from ${\tt Stab}_\pi$
%(because of the condition ``there exists walk $\tau$ with $\tau$ with $\tau\myeq W$''
%present in the definitions of both ${\tt Stab}_\pi$ and ${\tt Stab}^{\tt rev}_\pi$).
%Thus, Theorem~\ref{th:StabCounting:rev} does not automatically follow from Theorem~\ref{th:StabCounting}.
%Their proofs, however, are very similar (see Section~\ref{sec:proof:StabCounting}).

%For a word $W=w_1\ldots w_t$ we denote ${\tt REV}[W]=w_t\ldots w_1$ to be its reverse.
%We will deal with walks $\tau\myeq W$
%for which either $W$ is a stable word or ${\tt REV}[W]$ is a stable word.
%Next, we observe that sequences corresponding to such walks satisfy an additional property.
We say that a sequence $\varphi=(I_1,\ldots,I_s)$ with $s\ge 1$ is {\em strongly stable}
if \\
(i) $I_r\in\Ind(F)$ for each $r\in[s]$, \\
(ii) $I_{r+1}\subseteq \Gamma(I_r)$ for each $r\in[s-1]$, and \\
(iii) $I_r\ne\varnothing$ for each $r\in[2,s]$. \\
(Compared to Definition~\ref{def:StableSequence}, we added condition (iii), and in (ii) replaced condition $I_{r+1}\subseteq \Gamma^+(I_r)$ with a stronger condition
$I_{r+1}\subseteq \Gamma(I_r)$). 
%The proposition below implies that it suffices to count strongly stable sequences instead of stable sequences. 
%
For a stable word $W$ let $\varphi_W=(I_1,\ldots,I_s)$ be the corresponding stable sequence;
if $W$ is empty then $\varphi_W=(\varnothing)$.

\begin{proposition}\label{prop:StronglyStable}
(a) For any $W\in{\tt Stab}_\pi$ the sequence $\varphi_W$ is strongly stable.
(b) For any $W'\in{\tt Stab}^{\tt rev}_\pi$ the sequence $\varphi_{W}$ is strongly stable,
where $W$ is the reverse of word $W'$.
%Consider a walk $\tau\myeq W$. (a) If $W$ is a stable word then
%$\varphi_W$ is a strongly stable sequence. (b)
%If ${\tt REV}[W]$ is a stable word then $\varphi_{{\tt REV}[W]}$ is a strongly stable sequence.
%If $\tau$ is a stable walk then $\varphi_\tau$ is a strongly stable sequence.
\end{proposition}
\begin{proof}
We need to show that the sequence $\varphi_W=(I_1,\ldots,I_s)$ satisfies additionally property (ii).
We will prove this only in the case (a), i.e.\ under the assumption that there exists a walk $\tau\myeq W$.
The case (b), i.e.\ when there exists a walk $\tau\myeq W'$, is completely analogous.

Let $W=W_1\ldots W_s$ be the partitioning of $W$ given in Definition~\ref{def:Stable}.
It suffices to prove that if a flaw $f$ is present in adjacent segments $W_r$ and $W_{r+1}$
then $f\sim f$. Suppose not: $f\nsim f$.
Then by Lemma~\ref{lemma:technical}(b) there exists flaw $g$ between the two occurences of $f$
with $f\sim g$. We have $W_rW_{r+1}=\ldots f \ldots g \ldots f \ldots$.
We have either $g\in W_r$ or $g\in W_{r+1}$, and so we must have $f\nsim g$ - a contradiction.
\end{proof}

For a sequence $\varphi=(I_1,\ldots,I_s)$ with $s\ge 1$ let $R_\varphi=I_1$ be the first set in the sequence,
and  denote $|\varphi|=\sum_{r\in[s]}|I_r|$ and  $\lambda_\varphi=\prod_{r\in[s]}\prod_{f\in I_r} \lambda_f$.
Let {\tt Stab} be the set of strongly stable sequences,  ${\tt Stab}(R)=\{\varphi\in{\tt Stab}\::\:R_\varphi=R\}$,
and ${\tt Stab}(R, t)=\{\varphi\in{\tt Stab}(R)\::\:|\varphi|\ge t\}$.
%and ${\tt Stab}(R, t)={\tt Stab}(R)\cap{\tt Stab}(t)$.

Clearly, a $\pi$-stable word $W$ can be uniquely reconstructed from the corresponding sequence $\varphi_W$.
Thus, $W\mapsto \varphi_W$ is an injective mapping from ${\tt Stab}_\pi(R,t)$ to ${\tt Stab}(R,t)$.
Also, $\lambda_W=\lambda_{\varphi_W}$ for any $W\in{\tt Stab}_\pi(R,t)$.
This means that Theorem~\ref{th:StabCounting} will follow from the following result.

\begin{theorem}\label{th:StabCounting'}
Suppose that $(\rho,\sim)$ satisfies either the cluster expansion condition~\eqref{eq:Condition:b}
or the Shearer's condition from Definition~\ref{def:Shearer}. \footnote{Note that Bissacot et al.~\cite{Bissacot}
proved that the cluster expansion condition implies Shearer's condition, so it would suffice
to prove just the second claim. However, as mentioned in Remark~\ref{remark:cluster-expansion}, the definition of the cluster expansion condition
was slightly stronger than condition~\eqref{eq:Condition:b}.
Due to this annoying technicality we consider the two cases separately in the proof.
%(eventhough in many applications the two conditions coincide, since usually $f\in\Gamma(f)$).
}
Then
\begin{equation}\label{eq:StabCounting'}
\sum_{\varphi\in{\tt Stab}(R,t)} \lambda_\varphi \le \mu(R) %\frac{q_R(p)}{q_\varnothing(p)}
\cdot \theta^{t} \qquad \forall R \in \Ind(F)
\end{equation}
\end{theorem}
\begin{proof}
First, assume that the Shearer's condition  holds.
In this case the claim has been proven in~\cite{KolipakaSzegedy,HarveyVondrak15}.
To elaborate, let $p\in\mathbb R^{|F|}$ be the vector from Definition~\ref{def:Shearer},
and define
\begin{eqnarray}
%\lambda_\varphi&=&\prod_{r\in[s]}\prod_{f\in I_r} \frac{1}{A_f} %\\
p_\varphi&=&\prod_{r\in[s]}\prod_{f\in I_r} p_f
\end{eqnarray}
%From definitions we have $\lambda_\varphi\le p_\varphi\cdot\theta^t$ for any $\varphi\in {\tt Stab}(R,t)$. 
Let ${\tt Stab}'$ be the set of sequences $\phi=(I_1,\ldots,I_s)$ that satisfy conditions (i) and (ii) given
in the beginning of this section. For any integer $\ell$
we can write
\begin{equation}\label{eq:FADSGAHSIA}
\sum_{\substack{\varphi\in{\tt Stab}(R,t) \\ \mbox{\small length of $\varphi$ is at most $\ell$}}} p_\varphi \le \sum_{\varphi=(R,I_2,\ldots,I_\ell)\in{\tt Stab}'} p_\varphi \le
\frac{q_R(p)}{q_\varnothing(p)}
%\cdot \theta^{t} 
\qquad \forall R \in \Ind(F)
\end{equation}
where the first inequality holds since for any $\varphi=(I_1,\ldots,I_s)\in{\tt Stab}(R,t)$ with $s\le \ell$
there is a corresponding sequence $\varphi'=(R,I_2,\ldots,I_s,\varnothing,\ldots,\varnothing)\in{\tt Stab}'$ of length $\ell$,
and the second inequality appears implicitly eq. (2) in~\cite{KolipakaSzegedy} and as Lemma 5.10 in \cite{HarveyVondrak15}.
%Note that the result in~\cite{KolipakaSzegedy,HarveyVondrak15} is formulated
%for stable sequences. However, restricting to strongly stable sequences
%can only make the sum
%$\sum_{\varphi\in{\tt Stab}(R,t)} p_\varphi$ smaller, and so~\eqref{eq:FADSGAHSIA} also holds for strongly stable sequences.
Taking a limit $\ell\rightarrow\infty$ in~\eqref{eq:FADSGAHSIA}
and observing that  $\mu(R)=\frac{q_R(p)}{q_\varnothing(p)}$, we get
$\sum_{\varphi\in{\tt Stab}(R,t)} p_\varphi \le \mu(R)$.
From definitions we have $\lambda_\varphi\le p_\varphi\cdot\theta^t$ for any $\varphi\in {\tt Stab}(R,t)$,
this gives~\eqref{eq:StabCounting'}.

Now assume that condition~\eqref{eq:Condition} holds.
We say that a pair of subsets $(R,S)$ is {\em independent} if $R\cap S=\varnothing$ and $R\cup S \in \Ind(F)$.
For such pair let ${\tt Stab}(R,S,t,\ell)$ be the set of sequences of the form $\varphi=(R,I_2,\ldots,I_s)$
with $s\le \ell$, $|\varphi|\ge t$
that satisfy one of the following: 
\begin{itemize}
\item $|s|=1$ (i.e. $\varphi=(R)$) and $S=\varnothing$;
\item $|s|\ge 2$, $S= I_2-\Gamma(R)$ and $(I_2,\ldots,I_s)\in{\tt Stab}$.
%(i) $R=I_1$, 
%(ii) $S= I_2-\Gamma(I_1)$, and 
%(iii) $I_{r+1}\subseteq\Gamma(I_r)$ for any $r\in[2,s-1]$. 
\end{itemize}
%Here we assumed that $I_r$ is empty for $r>s$.
It can be seen that ${\tt Stab}(R,\varnothing,t,\infty)={\tt Stab}(R,t)$
for an independent set $R$.
%
%$\varphi\in{\tt Stab}(R,t)$
%such that $S$ (if non-empty) belongs to the second set of $\varphi$. 
%

We say that a tuple $(R,S,t,\ell)$ is {\em valid} if $(R,S)$ is an independent pair
and $t\ge 0$, $\ell\ge 1$ are integers.
We will prove the following for any valid tuple $(R,S,t,\ell)$:
\begin{equation}\label{eq:ALKDSFNAKSJFA}
\sum_{\varphi\in{\tt Stab}(R,S,t,\ell)} \lambda_\varphi \le \mu(R)\mu(S)\cdot \theta^{t} 
%\qquad \forall (R,S)\mbox{ s.t.\ } R\cap S=\varnothing, R\cup S \in \Ind(F)
\end{equation}
%(Here we assume that ${\tt Stab}(\varnothing,\varnothing,0)$ contains a single sequence, namely the empty sequence).
%
Let us introduce a partial order $\sqsubseteq$ on tuples $(R,S,t,\ell)$
as the lexicographical order on vectors $(\ell,|R|)$ (the first component is more significant).
We use induction on this partial order. The base case is given by a tuple $(R,S,t,1)$.
We can assume that $S=\varnothing$ and $t\le |R|$ (otherwise ${\tt Stab}(R,S,t,1)$ is empty).
In this case ${\tt Stab}(R,S,t,1)$ contains a single sequence $\varphi=(R)$, and so
$$
\sum_{\varphi\in{\tt Stab}(R,\varnothing,t,1)}\lambda_\varphi\;=\; \prod_{f\in R}\lambda_f \;\le\; \prod_{f\in R} (\mu_f\cdot\theta)=\mu(R)\cdot\theta^{|R|} \;\le\;\mu(R)\cdot\theta^t
$$
where we used inequality $\lambda_f\le\mu_f\cdot\theta$ that follows from condition~\eqref{eq:Condition:b} with $S=\varnothing$.

Now consider a valid tuple $(R,S,t,\ell)$ with $\ell\ge 2$, and assume that the claim holds for lower tuples. Two cases are possible.
% Assume that $(R,S) \ne(\varnothing,\varnothing)$ (otherwise ${\tt Stab}(R,S,t)$
%is empty). 

\begin{itemize}
\item {$R=\varnothing$.} We have a natural isomorphism between sets ${\tt Stab}(\varnothing,S,t,\ell)$
and ${\tt Stab}(S,\varnothing,t,\ell-1)$, namely ${\tt Stab}(\varnothing,S,t,\ell)=\{(\varnothing,I_1,\ldots,I_r)\:|\:(I_1,\ldots,I_r)\in
{\tt Stab}(S,\varnothing,t,\ell-1)\}$. This gives
$$
\sum_{\varphi\in{\tt Stab}(\varnothing,S,t,\ell)}\lambda_\varphi\;=\; \sum_{\varphi\in{\tt Stab}(S,\varnothing,t,\ell-1)}\lambda_\varphi\;\le\;\mu(S)\cdot \theta^t
$$
where in the last inequality we used the induction hypothesis.

\item {$R\ne\varnothing$.}
%Therefore, it suffices to prove the claim in the case when  $R\ne\varnothing$.
Pick $f\in R$, and denote $R^-=R-\{f\}$. Let $\tilde\Gamma(f)$ be the set of flaws $g\in \Gamma(f)$
that (i) can occur in the second set of a sequence $\varphi\in {\tt Stab}(R,S,t,\ell)$, and (ii) do belong to $\Gamma(f')$ for any $f'\in R^-$.
Formally, $\tilde\Gamma(f)=\{g\in \Gamma(f)\:|\:g\nsim f'\mbox{ for all } f'\in R\cup S\}$.
We can write
\begin{eqnarray*}
\sum_{\varphi\in{\tt Stab}(R,S,t,\ell)} \lambda_\varphi 
&=&  \sum_{T\in\Ind(\tilde\Gamma(f))} \sum_{\varphi\in{\tt Stab}(R^-,S\cup T,t-1,\ell)} \lambda_f\cdot\lambda_\varphi \\
&\stackrel{\mbox{\tiny(a)}}\le& \lambda_f \sum_{T\in\Ind(\tilde\Gamma(f))} \mu(R^-)\mu(S\cup T)\cdot \theta^{t-1} \\
&\stackrel{\mbox{\tiny(b)}}= &  \mu(R)\mu( S)\theta^{t-1}\cdot \frac{\lambda_f}{\mu_f} \sum_{T\in\Ind(\tilde\Gamma(f))} \mu(T) \\
&\stackrel{\mbox{\tiny(c)}}\le& \mu(R)\mu( S)\theta^{t-1}\cdot \theta
\end{eqnarray*}
where (a) is by the induction hypothesis, (b) is true since $\mu(R^-)=\mu(R)/\mu_f$ and $\mu(S\cup T)=\mu(S)\mu(T)$,
and (c) follows from condition~\eqref{eq:Condition:b}. 
%This establishes the induction step.
\end{itemize}

\end{proof}

\section*{Acknowledgments}
The author thanks Dimitris Achlioptas and anonymous reviewers for useful suggestions that helped to improve the presentation of the paper.
This work was supported by the European Research Council under the European Unions Seventh Framework Programme (FP7/2007-2013)/ERC grant agreement no 616160.

%\section*{Acknowledgments}
%We would like to acknowledge the assistance of volunteers in putting
%together this example manuscript and supplement.

\bibliographystyle{plain}
\bibliography{LLL}
\end{document}